\documentclass[11pt,letterpaper]{article}
\usepackage[T1]{fontenc}
\usepackage[utf8]{inputenc}
\usepackage{lmodern}
\usepackage[margin=1in]{geometry}
\usepackage{microtype}
\usepackage{amsmath,amsthm,amssymb}
\usepackage{bbm}
\usepackage{graphicx}
\usepackage{subcaption}
\usepackage{natbib}
\usepackage[hidelinks,hypertexnames=false]{hyperref}
\usepackage{cleveref}
\hypersetup{
  pdftitle={Robust Mechanism Design with Anonymous Information},
  pdfauthor={Zhihao Gavin Tang and Shixin Wang}
}

\newtheorem{proposition}{Proposition}[section]
\newtheorem{definition}{Definition}[section]
\newtheorem{theorem}{Theorem}[section]
\newtheorem{lemma}{Lemma}[section]
\newtheorem{claim}[theorem]{Claim}
\newtheorem{fact}{Fact}[section]

\DeclareMathOperator*{\argmax}{arg\,max}

\newcommand{\ind}{\mathbbm{1}}
\providecommand{\E}{\mathbb{E}}
\providecommand{\cT}{\mathcal{T}}
\newcommand{\off}[1]{\left[#1\right]}
\newcommand{\offf}[1]{\left\{#1\right\}}
\newcommand{\Bigof}[1]{\Big(#1\Big)}
\providecommand{\Rev}{\mathrm{Rev}}
\providecommand{\OPT}{\mathsf{OPT}}
\providecommand{\opt}{\mathsf{OPT}}
\newcommand{\F}{\mathbb{F}}
\newcommand{\G}{\mathbb{G}}
\newcommand{\dists}{\mathbf{F}}
\newcommand{\OS}[1][k]{\Phi_{{#1}}}
\newcommand{\Fiid}[1][k]{\bar{\F}_{{#1}}}
\newcommand{\vals}{\mathbf{v}}
\newcommand{\mech}{\mathcal{M}}
\newcommand{\R}[1][k]{\mathcal{R}_{#1}}
\newcommand{\spa}{\mathsf{SPA}}
\newcommand{\pp}{\mathsf{PP}}
\newcommand{\myerson}{\mathsf{Myerson}}
\newcommand{\Ex}[2][]{\underset{{#1}}{\E}\off{{#2}}}
\renewcommand{\P}{\mathbb{P}}
\renewcommand{\Pr}[2][]{\underset{{#1}}{\P}\off{{#2}}}
\newcommand{\dd}[1]{\;\textrm{d}#1}

\title{Robust Mechanism Design with Anonymous Information}
\author{
  Zhihao Gavin Tang\thanks{\raggedright Key Laboratory of Interdisciplinary Research of Computation and Economics, Shanghai University of Finance and Economics. \texttt{tang.zhihao@mail.shufe.edu.cn}}
  \and
  Shixin Wang\thanks{\raggedright H. Milton Stewart School of Industrial and Systems Engineering, Georgia Institute of Technology. \texttt{shixin.wang@isye.gatech.edu}}
}
\date{}

\begin{document}
\maketitle

\begin{abstract}
  In practice, auction data are often endogenously censored and anonymous, revealing only limited outcome statistics rather than full bid profiles.
  We study robust auction design when the seller observes only the distribution of a single anonymous order statistic of the bidders' values and seeks to maximize worst-case expected revenue over all product distributions consistent with the observed statistic. We show that simple and widely used mechanisms are robustly optimal. Specifically, posted pricing is robustly optimal given the distribution of the highest value; the Myerson auction designed for the unique consistent i.i.d.\ distribution, with uniformly random tie-breaking, is robustly optimal given the lowest value distribution; and the second-price auction with an optimal reserve is robustly optimal when an intermediate order statistic is observed and the implied i.i.d.\ distribution is regular above its reserve. Without any regularity assumption, when the second-highest value is observed, the best worst-case revenue of a second-price auction with a reserve is still at least a $2-\sqrt{2}\approx 0.586$ fraction of the optimal worst-case revenue. More generally,
 for a broad class of monotone symmetric mechanisms depending only on the top $k$ order statistics, including multi-unit and position auctions, the worst-case revenue is attained under the i.i.d.\ distribution consistent with the observed $k$-th order statistic.
  Our results provide a tractable foundation for non-discriminatory auction design, where fairness and privacy are intrinsic consequences of the information structure rather than imposed constraints.
\end{abstract}

\section{Introduction}

Auction mechanisms are widely used to allocate goods and resources in practice, ranging from online advertising and e-commerce platforms to procurement, real estate, and spectrum sales.
The seminal work \cite{myerson1981optimal} characterizes revenue-optimal mechanisms under the assumption that the seller has precise knowledge of each bidder’s value distribution.
While this offers an elegant theoretical framework, in many real-world markets bidder-level information is often unavailable or unusable: auction outcomes are endogenously censored by auction formats, bidder identities are anonymous, and privacy and fairness considerations restrict the use of discriminatory mechanisms. Consequently, transaction data are often anonymous and record only order statistics of the bids rather than full bid profiles.
In this work, we study a robust mechanism design problem in which the seller possesses only aggregated and anonymous information about buyers’ distributions, motivated by the following two fundamental features of real-world auction markets.

\paragraph{Limited Observability on Bidder-Level Data.}
First, in practice, sellers rarely observe bidder-level valuations when learning from historical transactions or public disclosures; instead, only aggregated and anonymous information is typically available and reliable due to endogenous censoring of observed auction outcomes and the anonymity and scarcity of bidder data.

\textit{Inherently censored auction outcomes.}
When sellers hope to collect bidders' data from historical transactions, the auction outcomes observed by the sellers are often heavily censored. Many widely used auction formats only elicit some order statistics information inherently.
When a seller runs a posted-price mechanism, for instance,  he observes only a binary purchase decision, revealing whether the highest bidder's valuation exceeds a threshold, while the specific values of bidders remain entirely hidden.
In ascending (English) auctions, the winning bidder never reveals their true maximum valuation, and only a bid marginally higher than the second-highest bid is observed. Furthermore, losing bidders typically exit the market silently once the price exceeds their valuation, and their specific valuations are often unrecorded. Similarly, ad exchanges and procurement platforms often disclose only anonymous winning prices, which is the second-highest bid.

\textit{Anonymous and sparse bidder data.} Even when richer data might be available in principle, bidder-level learning remains challenging. In many markets, buyers participate infrequently, and they may use multiple or temporary identifiers (e.g., false-name bidding, shill bidding). These behaviors introduce substantial noise, sparsity, and instability into individual-level bid observations, limiting their reliability for distributional inference, whereas winning prices in the transactions, which are often order statistics of bidders' valuations, are substantially more stable and reliable.

\paragraph{Regulated Usage of Bidder-Level Data.}
Second, the growing emphasis on privacy and anonymity limits the practical applicability of discriminatory mechanisms, such as the optimal Myerson auction.
In particular, when buyers’ prior distributions are not i.i.d., Myerson’s auction may allocate the item to a bidder with a lower bid, raising concerns about fairness, transparency and acceptability in practice. Much of the existing literature addresses these concerns by imposing symmetry, anonymity, or fairness as explicit constraints in the mechanism design problem. In contrast, we take an information-based approach and model the seller’s limited access to bidder-level data directly through a robust mechanism design framework grounded in aggregated information.

\textit{Legal constraints for privacy and fairness.} This perspective is motivated by the growing difficulty of obtaining or exploiting bidder-specific information in modern auction environments.
Privacy regulations such as the GDPR\footnote{General Data Protection Regulation \url{https://gdpr-info.eu}} and CCPA\footnote{California Consumer Privacy Act \url{https://oag.ca.gov/privacy/ccpa}}, alongside platform-level data minimization standards, frequently restrict sellers to observing or using only aggregated or anonymized transaction records. Anonymous data protects individual identities, avoids misuse of sensitive bidder information, and prevents discrimination of individuals based on their identities and past behaviors. As a result, platforms often release or rely on only coarse, aggregated information, such as final payments or clearing prices, rather than full bid profiles.

\textit{Strategic motivations for anonymity.}
Beyond regulatory considerations, anonymity is also a deliberate design choice in many auction platforms. In ad exchanges, procurement platforms, eBay, and real estate auctions, buyer identities are often intentionally hidden from sellers and other participants. They typically disclose only anonymous outcomes, such as highest bid or winner payments, in order to mitigate strategic behaviors, including collusion, manipulation, and discriminatory pricing.

Our framework directly addresses these privacy and anonymity considerations without restricting the mechanism space.
Specifically, we formulate a robust mechanism design framework in which the seller observes only aggregated order-statistic information about bidders’ valuations and imposes no ex ante constraints on the form of the mechanism. This framework protects data privacy instead of restricting mechanisms. Because the available information is symmetric across bidders, randomly relabeling the bidders does not lower the worst-case revenue of any mechanism, so robust mechanisms can be taken to be anonymous and non-discriminatory without loss. Thus, our model naturally addresses fairness and privacy considerations not as external constraints but as intrinsic features of the information structure.

\subsection{Main Results}

Motivated by the limited availability of bidder-level distributional information and the widespread use of anonymous, order-statistic data in practice, we study auction design when the seller observes the distribution of an order statistic of bidders’ values. Specifically, we develop a robust mechanism design framework in which the seller observes, for example, the highest or second-highest bid—but lacks knowledge of the underlying individual value distributions. This information structure reflects common empirical settings in which platforms can reliably estimate the distribution of transaction prices or winning bids from past auctions, yet cannot recover bidder-specific valuations.
Within this framework, the seller faces an ambiguity set consisting of all product distributions consistent with the observed order-statistic distribution and seeks to maximize worst-case expected revenue across this entire set.
We develop a systematic treatment of robust mechanism design across different types of order-statistic information. Our results show that the widespread use of simple mechanisms such as posted pricing and second price auctions in practice is supported by their robust performance under censored and aggregated data.

\begin{theorem}[Informal]
When the distribution of the first order statistic is known, posting a take-it-or-leave-it price optimized for this distribution is robustly optimal.
\end{theorem}

We begin with the natural setting in which the seller knows the distribution of the first order statistic, that is, the highest valuation. Since the distribution of the first order statistic can be directly learned from repeated pricing experiments, a natural candidate mechanism we consider in this case is a posted-price mechanism. Intuitively, when only the distribution of the highest valuation is identified, posted pricing aligns closely with the underlying data-generating process and is therefore inherently robust to uncertainty about the remaining aspects of the value distribution.

Moreover, the ambiguity set under consideration includes a degenerate product distribution consisting of one active bidder and $n-1$ dummy bidders with identically zero valuations. The presence of this distribution certifies the optimality of simple posted-price mechanisms.

\begin{theorem}[Informal]
When the distribution of the last order statistic is known, Myerson’s auction designed for the consistent i.i.d.\ distribution, with uniformly random tie-breaking, is robustly optimal.
\end{theorem}

Next, we study the other extreme when the distribution of the last order statistic is known. At a high level, this case is guided by a similar underlying intuition as the first-order statistic setting: robust optimality can be characterized by identifying a least favorable distribution within the ambiguity set and evaluating the performance of an optimal mechanism against it.

We show that Myerson’s auction designed for the unique i.i.d.\ distribution consistent with the known order-statistic distribution is robustly optimal. The key step is to show that this i.i.d.\ distribution is least favorable for the auction when ties are broken uniformly at random. We express the auction's revenue as an integral, over bid thresholds, of the probabilities of the events that determine its critical payments, and we prove a two-threshold comparison lemma showing that each such probability is minimized when all bidders are identically distributed.

\begin{theorem}[Informal]
When the distribution of the second order statistic is known, the second-price auction with a reserve designed for the consistent i.i.d.\ distribution is robust and conditionally optimal, and its best robust revenue over all reserves is always at least a $2-\sqrt{2}\approx0.586$ fraction of the optimal robust revenue. In contrast, Myerson’s auction designed for the same consistent i.i.d.\ distribution is not robust when ironing is required.
\end{theorem}

Then we turn to the most practically relevant setting in which the distribution of the second order statistic is known. This type of information can be naturally collected by a seller through ascending auctions and second price auctions, or observed by an outside third party from final payment disclosures.
This setting is substantially more challenging: we show that, if $n\ge3$ and $\G$ is not a point mass, no distribution in the associated ambiguity set minimizes all order statistics simultaneously. Consequently, the approach used for the first and last order statistics, which compares all mechanisms against a single least favorable distribution, does not directly extend.

Despite this difficulty, we show that the i.i.d.\ distribution consistent with the observed distribution of the second-highest value minimizes the highest value in the sense of first-order stochastic dominance. Leveraging this property, we characterize the performance guarantees of a widely used mechanism, the second-price auction with a reserve, within this ambiguity set. In particular, we prove that if the implied i.i.d.\ distribution is regular above its optimal reserve, then the second-price auction with the corresponding reserve is robustly optimal among all distributions consistent with the given second-highest-value distribution.

This result has a natural practical interpretation. When a seller observes final payments from repeated second-price or ascending auctions, he can estimate the distribution of the second-highest value and compute the optimal reserve under the i.i.d.\ assumption. Our analysis shows that introducing this reserve yields a uniform revenue improvement across all distributions consistent with the observed second-highest-value distribution.
Our findings complement the classical result of \cite{BulowK1996}, which shows that when bidders are i.i.d.\ and fully regular, the revenue gain from introducing a reserve price in a second-price auction becomes negligible as the number of bidders grows large. In contrast, we show that when bidders are i.i.d.\ but regular only above the optimal reserve---for example, when a regular bidder appears only with some probability---introducing a reserve can yield a constant-factor revenue improvement.

Moreover, we show that Myerson’s auction designed for the consistent i.i.d.\ distribution is not robust. In particular, its revenue on the consistent i.i.d.\ distribution only provides an upper bound on the worst-case revenue of the robust optimization problem. We further present concrete examples in which the second-price auction with a reserve is robustly optimal, whereas Myerson’s auction is not.

Although our robust optimality guarantee for the second-price auction with a reserve requires regularity above the reserve, we show that the auction is always approximately optimal. For every number of bidders and every distribution of the second-highest value, without any regularity, continuity, or moment assumptions, the best robust revenue of the second-price auction over all reserves is at least a $2-\sqrt{2}\approx 0.586$ fraction of the optimal worst-case revenue. Hence a suitable reserve secures this fraction up to an arbitrarily small loss (with no loss when the distribution has bounded support), and it can be computed from the consistent i.i.d.\ distribution alone. The proof compares both sides with the optimal posted-price revenue $P$ under the consistent i.i.d.\ distribution: with a reserve equal to any posted price, the second-price auction earns at least that price's i.i.d.\ posted-price revenue under every consistent distribution, whereas on a consistent instance with only two active bidders, a scalar inequality shows that no mechanism earns more than $(1+1/\sqrt{2})P$.

Finally, we extend our analysis to a general setting in which the seller observes the distribution of an arbitrary intermediate order statistic. This setting naturally arises in multi-unit auctions with unit-demand bidders. For example, the seller could observe the final payment after selling $k$ identical units to unit-demand bidders by a multi-unit auction, whose payment is the $(k+1)$-th order statistic. In practice, sellers can reliably observe the distribution of such payments from historical transactions and may seek to leverage this information to improve revenue by incorporating reserve prices. We prove that for any reserve, the worst-case revenue is attained under the i.i.d.\ distribution consistent with the observed distribution of the $(k+1)$-th order statistic. Therefore, the robustly optimal reserve can be characterized by solving the seller’s problem under the i.i.d.\ benchmark. More generally, we establish that for any mechanism whose total payment is separable and nondecreasing in the first $k+1$ order statistics, the worst-case revenue over all distributions consistent with the observed $(k+1)$-th order statistic is achieved by the corresponding i.i.d.\ distribution. These results extend beyond standard multi-unit auctions and apply to position auctions, which are widely used in sponsored search and online advertising platforms.

Taken together, our results propose a unified framework for robust auction design under order-statistic information. By focusing on anonymous, aggregated statistics that naturally arise from auction outcomes, our approach aligns closely with the informational constraints imposed by privacy regulation and platform design. Moreover, because the seller’s information is symmetric across bidders, the resulting mechanisms are inherently anonymous, fair, and non-discriminatory, without the need to impose such properties as explicit constraints. This framework therefore provides both a tractable theoretical foundation and a practical design principle for auction environments in which detailed bidder-level data are unavailable or unusable.

\subsection{Literature Review}
\paragraph{Sequential posted-price mechanisms.}
Algorithms that rely only on the distribution of the first order statistic have been extensively studied in the context of sequential posted-price mechanisms, primarily through the prophet inequality framework~\citep{krengel1977semiamarts,krengel1978semiamarts}.
In this setting, prices are offered sequentially to arriving buyers, and performance is evaluated via competitive analysis against the optimal offline benchmark for social welfare, rather than revenue.\footnote{When the full distributional information is known, \cite{aaai/HajiaghayiKS07} and \citet{correa2019pricing} establish an equivalence between the objectives of maximizing social welfare and maximizing revenue.}

When buyers arrive in an adversarial order, it is known that optimal competitive ratios can be achieved by fixed pricing rules, such as posting the median of the maximum value or half of its expectation. Notably, these mechanisms depend only on the distribution of the first order statistic and are therefore inherently symmetric across buyers. In the prophet secretary setting with random arrival order, a family of blind strategies has been developed \citep{mp/CorreaSZ21}, which again relies solely on the distribution of the largest value.
Beyond these results, \cite{ec/EzraFT24} directly incorporates identity-blindness as an explicit design constraint, studying mechanisms whose pricing rules are blind to buyer identities.

\paragraph{Robust auction.}
Robust auction design has been studied primarily under i.i.d.\ buyers, where robustly optimal or approximately optimal mechanisms can be characterized under regularity or MHR conditions \citep{allouah2020prior,fu2015randomization,corr/HartlineJL25}, as well as under partial information about the common distribution, such as known means and upper bounds \citep{suzdaltsev2022distributionally,bachrach2022distributional} or known supports \citep{anunrojwong2022robustness,anunrojwong2023robust}.
Beyond i.i.d.\ distributions, a related literature studies robust mechanisms given marginal information while allowing for arbitrary correlation across bidders \citep{bei2019correlation,he2022correlation,koccyiugit2020distributionally,koccyiugit2024regret}.
We complement this literature by studying an intermediate and practically relevant setting with heterogeneous but independent buyers. By incorporating order-statistic information into the ambiguity set, we characterize robustly optimal mechanisms in this setting.

\paragraph{Auctions with nondiscrimination and privacy considerations.}
Symmetric auction formats are appealing because they are anonymous and non-discriminatory. \cite{deb2017discrimination} characterize the set of outcomes implementable by symmetric auctions and show that, despite their anonymity, such mechanisms can endogenously generate highly discriminatory allocations—including the revenue-optimal auction—through bidders’ strategic behavior. \cite{eilat2023privacy} quantify privacy loss in auctions and demonstrate that the second-price auction minimizes privacy loss among all DSIC and ex-post IR mechanisms. \cite{chen2025optimal} study second-price auctions with flexible reserve prices and show that they are revenue-maximizing among all discrimination-free auctions. In contrast to approaches that impose privacy or non-discrimination as explicit design constraints, we propose a robust auction design framework based on anonymous, aggregated data, in which non-discrimination arises endogenously from the underlying information structure.

\section{Preliminaries}

Consider a single-item auction with $n$ buyers. Each buyer has a private valuation, and the valuation profile $\vals = (v_1,\dots,v_n)$ is drawn from a product distribution $\dists = \F_1 \times \cdots \times \F_n$.
Throughout, valuations are nonnegative, i.e., every $\F_i$ is supported on $[0,\infty)$. We use the same symbol for a distribution and its distribution function, and write $\F(v^-):=\lim_{u\uparrow v}\F(u)=\Pr[v'\sim\F]{v'<v}$.

For a valuation profile $\vals \in \mathbb{R}^n$, let $v_{(k)}$ denote the $k$-th order statistic, that is, the $k$-th largest value among $\{v_1,\dots,v_n\}$. We denote $\OS(\dists)$ as the marginal distribution of $v_{(k)}$ when $\vals \sim \dists$.

The seller does not observe the full joint distribution $\dists$. Instead, the seller observes only an aggregated statistic---namely, the distribution $\G$ of the $k$-th order statistic induced by $\dists$. Equivalently, the seller knows that the underlying distribution belongs to the ambiguity set
\[
\Pi_k(\G) := \{ \dists : \OS(\dists) = \G \}~.
\]

The seller aims to design a mechanism $\mech$ that maximizes the worst-case expected revenue over all distributions in the ambiguity set $\Pi_k(\G)$:
\[
\R(\G) := \sup_{\mech} \inf_{\dists \in \Pi_k(\G)} \Rev(\mech, \dists)~,
\]
where $\Rev(\mech,\dists)$ denotes the expected revenue of the mechanism $\mech$ when buyers’ valuations are drawn from the distribution $\dists$.
The supremum is taken over all (possibly randomized) single-item mechanisms that are dominant-strategy incentive compatible (DSIC) and ex-post individually rational (IR); for randomized mechanisms, both conditions are imposed on expected allocations and payments conditional on the reported profile. These quantities may be infinite, and we write $\max$ and $\min$ when the extrema are attained. Since Myerson's optimal auction is DSIC, the optimal revenue $\opt(\dists)$ defined below is also the largest revenue that a DSIC and IR mechanism can obtain under a known distribution $\dists$.

For two distributions $\F$ and $\F'$ on $\mathbb{R}$, we write $\F\preceq\F'$, or equivalently $\F'\succeq\F$, if $\F'$ first-order stochastically dominates $\F$, that is, $\F(v)\ge\F'(v)$ for every $v\in\mathbb{R}$. In this case, $\Pr[v\sim\F]{v\ge t}\le\Pr[v\sim\F']{v\ge t}$ for every $t\in\mathbb{R}$, and $\Ex[v\sim\F]{f(v)}\le\Ex[v\sim\F']{f(v)}$ for every nondecreasing function $f$ that is bounded from below. We only compare distributions of individual order statistics in this sense.

We denote by $\Fiid(\G)$ the common marginal of the i.i.d.\ distribution whose $k$-th order statistic has distribution $\G$; it depends on $n$, which is fixed unless stated otherwise.
Formally, for every $k \in [n]$ and any $\G$, we define $\Fiid(\G)$ to be the unique distribution satisfying
\[
\OS\left( \Fiid(\G)^n \right) = \G .
\]
The existence and uniqueness of such a distribution are guaranteed by the following claim.
\begin{claim}\label{clm:iid-inversion}
Fix $n \ge 1$ and $k \in [n]$. For every distribution $\G$, there exists a unique distribution $\F$ such that
\[
\OS(\F^n) = \G, \qquad\text{where }\F^n = \F \times \cdots \times \F~.
\]
\end{claim}
\begin{proof}
Observe that $v_{(k)} \le v$ if and only if at most $k-1$ of the $n$ samples are greater than $v$. Hence, for every $v \in \mathbb{R}$,
\[
\Pr{v_{(k)} \le v} = \Pr{\# \{i : v_i > v\} \le k-1} = \sum_{t=0}^{k-1} \binom{n}{t}\, (1-\F(v))^t \cdot \F(v)^{n-t}.
\]
Define the polynomial map $H_{n,k}:[0,1]\to[0,1]$ by
\[
H_{n,k}(u) := \sum_{t=0}^{k-1} \binom{n}{t}\, (1-u)^t \cdot u^{n-t} = \Pr{\mathrm{Bin}(n,1-u) \le k-1}.
\]
Notice that $H_{n,k}$ is continuous and strictly increasing on $[0,1]$.
Moreover, $H_{n,k}(0)=0$ and $H_{n,k}(1)=1$, so $H_{n,k}$ is a bijection from $[0,1]$ to $[0,1]$ and admits a continuous and increasing inverse $H_{n,k}^{-1}$. Therefore, $\F(v) := H_{n,k}^{-1}(\G(v))$ is the unique distribution with the stated property.
\end{proof}

\subsection{Mechanisms}
To streamline the exposition in the subsequent analysis, we briefly review and formalize the auction mechanisms that play a central role in our results, including posted pricing, second-price auctions with a reserve, and Myerson’s auction.

\paragraph{Posted Pricing Mechanisms ($\pp$).}
A posted pricing mechanism posts a single price $p$ to all buyers. The item is sold to the first buyer whose valuation is at least $p$. Therefore, for any product distribution $\dists = \F_1 \times \cdots \times \F_n$, the expected revenue is given by
\[
\Rev(\pp(p), \dists) = p \cdot \Pr[\vals \sim \dists]{v_{(1)} \ge p} = p \cdot \Bigof{1 - \prod_{i=1}^n \F_i(p^-)}.
\]

\paragraph{Second-Price Auction with a Reserve ($\spa$).}
A second-price auction with a reserve price $p$ operates as follows. If every buyer's valuation is below $p$, the item is not sold. If exactly one buyer has valuation at least $p$, that buyer wins the item and pays the reserve price $p$. If at least two buyers have valuations at least $p$, the buyer with the highest bid wins and pays the second-highest bid.
When $p=0$, the mechanism reduces to the standard second-price auction. Accordingly, for any distribution $\dists$, the expected revenue of the second-price auction with reserve $p$ is
\begin{equation}
    \label{eqn:revenue_spa}
\Rev(\spa(p), \dists) = p \cdot \Pr[\vals \sim \dists]{v_{(1)} \ge p > v_{(2)}} + \Ex[\vals \sim \dists]{v_{(2)} \cdot \ind[v_{(2)} \ge p]}
\end{equation}
\paragraph{Myerson's Auction~\citep{myerson1981optimal} ($\myerson$).}

To describe Myerson's auction, we recall the definition of regular distributions and revenue curves.
A distribution $\F$ with pdf $f$ is called regular, if its virtual value function $\varphi(v) := v - \frac{1-\F(v)}{f(v)}$ is non-decreasing, or equivalently, the revenue curve $r(q) := q \cdot \F^{\text{-}1}(1-q)$ is concave.

When all marginal distributions $\F_i$'s are regular, Myerson's auction allocates the item to the buyer with highest virtual value (provided it is non-negative), and charges her the minimum winning bid, i.e., the smallest bid at which her virtual value weakly exceeds those of all other buyers (and 0).
Moreover, the cutoff point $p_i := \inf\{p : \varphi_i(p) \ge 0\}$ coincides with the optimal monopoly price $p_i^* := \argmax_p p \cdot (1-\F_i(p))$.

For irregular distributions, Myerson’s auction involves an additional ironing step.
Specifically, ironing replaces the revenue curve $r_i$ of distribution $\F_i$ by its concave envelope $\bar{r}_i$, which corresponds to flattening non-concave regions of $r_i$.
Equivalently, ironing pools together intervals of quantiles on which the virtual value function is non-monotone and assigns them a common ironed virtual value $\bar{\varphi}_i$.
The auction then allocates the item to the buyer with the highest ironed virtual value (provided it is non-negative), and charges the corresponding critical payment.

We omit the implementation details of the ironing procedure. Intuitively, the ironed virtual value function coincides with the virtual value function induced by the ironed distribution.
In Section~\ref{sec:last}, we only use two standard properties: the ironed virtual value function $\bar{\varphi}$ is nondecreasing, and allocating the item to a bidder with the highest nonnegative ironed virtual value, with ties broken independently of the bids, and charging critical payments is optimal~\citep{myerson1981optimal}, provided that bids outside the support of the distribution are treated as described in Section~\ref{sec:last}.
In Section~\ref{sec:example}, we present a concrete example that necessitates ironing, where the computation of the ironed virtual value function is explicitly illustrated.

\citet{myerson1981optimal} proved that this auction maximizes the expected revenue when the underlying distribution $\dists$ is known to the seller, namely,
\[
\Rev(\myerson(\dists), \dists) = \opt(\dists)~,
\]
where $\opt(\dists)$ denotes the optimal expected revenue achievable when buyers’ valuations are drawn from the distribution $\dists$.

In the single buyer setting, Myerson’s auction allocates the item if and only if the buyer’s (ironed) virtual value is non-negative. This rule is equivalent to a posted-price mechanism. Specifically, for every $\F \in \Delta(\mathbb{R})$,
\[
\myerson(\F) = \pp(p^*), \qquad \text{where } \; p^* \in \argmax_{p} \; p \cdot \bigl(1-\F(p^-)\bigr)~.
\]

In the multi-buyer setting, we focus on Myerson’s auction in the symmetric environment with $n$ i.i.d.\ buyers whose valuations are drawn from a common distribution $\F$. In this case, a single ironed virtual value function $\bar{\varphi}$ is applied symmetrically to all buyers.
Moreover, when ironing does not affect values above the monopoly reserve price, Myerson’s auction admits a simple implementation as a second-price auction with a reserve. This observation motivates the following definition.

\begin{definition}
Let $\F$ be a distribution with virtual value function $\varphi$, and let $p^* := \argmax\limits_p p \cdot (1-\F(p^-))$ be the monopoly reserve price.
We say that $\F$ is \emph{regular above the reserve} if $\varphi$ is non-decreasing on $[p^*,\infty)$.
\end{definition}

Equivalently, $\F$ satisfies regularity above the reserve if and only if the ironed revenue curve coincides with the original revenue curve on all quantiles corresponding to values $v \ge p^*$, that is, no ironing interval intersects $[p^*,\infty)$. For distributions without a density (e.g., distributions with atoms), we take this revenue-curve condition as the definition of regularity above the reserve.

\begin{fact}
\label{fact:spa}
Consider a single-item auction with $n$ i.i.d.\ bidders whose values are drawn from $\F$.
The second-price auction with a reserve price $p^*$ is optimal, i.e.,
\[
\Rev(\spa(p^*), \F^n) = \opt(\F^n)
\]
if $\F$ is regular above the reserve~\citep{myerson1981optimal}. Indeed, no ironing interval contains values above $p^*$, the ironed virtual value is nonpositive below $p^*$, and the allocation rule of $\spa(p^*)$, which assigns the item to a highest value at or above $p^*$, maximizes the ironed virtual surplus pointwise while being constant on every ironing interval.
\end{fact}

\section[First Order Statistic: k=1]{First Order Statistic: $k=1$}
We begin with the first order statistic to illustrate our framework, namely, the setting in which the distribution of the largest value $v_{(1)}$ is known.

We prove that a posted-price mechanism is optimal in this case.  Our proof proceeds by explicitly constructing a saddle point: we establish a revenue guarantee for the posted-price mechanism and exhibit a matching distribution under which no mechanism can achieve a higher revenue. The key intuition is that the first order statistic $v_{(1)}$ is naturally generated by posted-price mechanisms, as trade occurs if and only if the highest valuation is at least the posted price. This close alignment between the data-generating process and the mechanism itself renders posted pricing inherently robust to uncertainty about the remaining aspects of the value distribution. Although the analysis appears standard in hindsight, it is not \emph{a priori} obvious that the optimal solution admits such an asymmetric construction.

\begin{theorem}
\label{thm:first}
Let $p^*$ be an optimal monopoly price with respect to $\G$, i.e., $p^* \mathrel{\in} \argmax_{p} p \cdot (1-\G(p^-))$.\footnote{Such a price exists, for example, whenever $\G$ has bounded support, because $p\mapsto p\,(1-\G(p^-))$ is upper semicontinuous. In general, the proof shows that $\R[1](\G)=\sup_p p\,(1-\G(p^-))$, and this value is approached by posted prices.} Then, the posted pricing mechanism with price $p^*$ is robustly optimal:
\[
\min_{\dists \in \Pi_1(\G)} \Rev\left(\pp(p^*), \dists\right) = p^* \cdot \bigl(1-\G((p^*)^-)\bigr) = \R[1](\G)~.
\]
\end{theorem}
\begin{proof}
We prove the theorem by explicitly constructing a saddle point of the optimization problem.
For any distribution $\dists$ whose first order statistic satisfies $\OS[1](\dists)=\G$, we have
\[
\Rev\left(\pp(p^*), \dists\right)
= p^* \cdot \Pr[\vals \sim \dists]{v_{(1)} \ge p^*}
= p^* \cdot \bigl(1-\G((p^*)^-)\bigr),
\]
where the second equality follows from the fact that $v_{(1)}=\max \vals$ is distributed according to $\G$.
Hence, the posted-price mechanism $\pp(p^*)$ achieves the same revenue for all distributions that are consistent with the given first order statistic~$\G$.

To prove optimality, consider the degenerate instance in which there is a single buyer with value distribution~$\G$, together with $(n-1)$ dummy bidders whose values are identically zero, i.e.,
\[
\dists^* = \G \times \delta_{0} \times \cdots \times \delta_{0}~.
\]
Against this distribution, no mechanism can achieve revenue exceeding
\[
\sup_\mech \Rev(\mech, \dists^*) = \sup_{p} \; p \cdot \bigl(1-\G(p^-)\bigr) = p^* \cdot \bigl(1-\G((p^*)^-)\bigr).
\]
Indeed, by IR, the dummy bidders pay at most zero, and fixing their reports at zero turns any DSIC and IR mechanism into a DSIC and IR mechanism for a single buyer with value distribution $\G$, whose revenue is at most $\sup_p p\cdot\Pr[v\sim\G]{v\ge p}$~\citep{myerson1981optimal}.
Therefore, $p^*$ together with $\dists^*$ forms a saddle point, and the claimed optimality follows.
\end{proof}
Theorem~\ref{thm:first} effectively states that when the seller only knows the distribution of the highest value in the market, the robust solution collapses to treating the entire market as a single ``representative agent'' whose value is drawn from $\G$.
We next discuss several implications of \Cref{thm:first}, which highlight both the robustness of posted-price mechanisms and the scope of the theorem beyond the baseline setting.
\paragraph{Unknown Number of Bidders.}  The posted-price mechanism $\pp(p^*)$ is still robustly optimal when the number of bidders is unknown. This follows by allowing the adversary to add arbitrarily many dummy bidders whose values are deterministically zero.
\paragraph{Robustness to Correlation.} Although we assume independent private values throughout the paper, posted-price mechanisms are in fact robust to \emph{arbitrarily correlated} value distributions, as long as the distribution of the first order statistic is known.
\paragraph{Weak Revenue Improvements.} While posted pricing achieves the optimal worst-case revenue, a second-price auction (when feasible) with reserve price equal to $p^*$ can weakly improve the revenue guarantee. That is, it preserves the same worst-case guarantee, while potentially extracting strictly higher revenue for distributions that are not worst-case.

\section[Last Order Statistic: k=n]{Last Order Statistic: $k=n$}
\label{sec:last}

We now consider the case in which the seller knows the distribution
$\G$ of the smallest valuation. The natural benchmark is the unique
consistent i.i.d.\ distribution. Its common marginal is
\[
\F=\Fiid[n](\G),
\qquad
\F(v)=1-\bigl(1-\G(v)\bigr)^{1/n}.
\]
The ambiguity set $\Pi_n(\G)$ nevertheless allows heterogeneous
marginals.

Let $\bar\varphi$ denote the ironed virtual value function of $\F$, which is nondecreasing~\citep{myerson1981optimal}. We extend it to all bids in $[0,\infty)$ as follows: a bid outside the support of $\F$ receives the value of $\bar\varphi$ at the largest support point below it, and a bid below the support receives the value $-\infty$, i.e., is rejected. The extension is nondecreasing. It matters because critical payments depend on how bids outside the support are treated; with this extension, the critical payments defined below coincide with the payments of Myerson's optimal auction, so the mechanism attains $\opt(\F^n)$ under $\F^n$. We call $\bar\varphi(b)$ the \emph{score} of a bid $b$. Ties are broken by a \emph{priority order} $\pi$, a uniformly random ordering of the bidders drawn independently of the bids: among bidders with equal scores, the one ranked earlier in $\pi$ has priority. For a fixed $\pi$, each bidder's allocation is nondecreasing in her own bid. We use normalized critical payments: a losing bidder pays zero, and a winner pays the infimum of her winning bids, with the other bids and $\pi$ held fixed. The resulting mechanism is DSIC and ex-post IR for every realization of $\pi$, and its expected revenue includes the randomness of $\pi$.

\begin{theorem}\label{thm:last}
Let $\F=\Fiid[n](\G)$, and let $\mech=\myerson(\F^n)$ allocate the item
to a bidder with the highest nonnegative score, breaking
ties according to the uniformly random priority order, and charge normalized critical payments. Then $\mech$ is robustly optimal, with
\[
\min_{\dists\in\Pi_n(\G)}\Rev(\mech,\dists)
=\Rev(\mech,\F^n)
=\R[n](\G)\;=\;\opt(\F^n).
\]
\end{theorem}

Since $\F^n\in\Pi_n(\G)$, its optimal revenue is an upper bound on the
robust benchmark. The substantive step is to show that $\mech$ earns at
least this revenue under every consistent product distribution. We do so
by comparing its critical-payment events at each bid threshold and
integrating their probabilities.

\subsection{A Two-Threshold Comparison}
\label{sec:two-threshold}

A bidder's score may be below, equal to, or above a fixed
threshold. The following lemma represents these possibilities by labels
$0$, $1$, and $2$, with the random order representing the auction's
priority rule. The two product constraints will follow from the known
distribution of the minimum valuation.

\begin{lemma}[Two-threshold comparison]\label{lem:two-threshold}
Let $X_1,\ldots,X_n$ be independent random variables with
\[
\P[X_i=0]=1-a_i,\qquad \P[X_i=1]=a_i-b_i,\qquad \P[X_i=2]=b_i,
\]
where $0\le b_i\le a_i\le1$, and let $\pi$ be a uniformly random ordering of $[n]$, independent of $(X_1,\ldots,X_n)$. We call $X_i$ the \emph{label} of index $i$. Let $E$ be the event that some index with label $2$ is preceded in $\pi$ by another index with a nonzero label. Fix $0\le B\le A\le1$. Among all parameter vectors with $\prod_ia_i=A$ and $\prod_ib_i=B$, the probability $\P[E]$ is minimized by the symmetric vector with $a_i=A^{1/n}$ and $b_i=B^{1/n}$ for all $i$.
\end{lemma}

\begin{proof}

The case $n=1$ is immediate. Let $n\ge2$. We first show that replacing $(a_1,b_1)$ and $(a_2,b_2)$ by $(a,b)$, where $a:=\sqrt{a_1a_2}$ and $b:=\sqrt{b_1b_2}$, does not increase $\P[E]$. This replacement preserves both products and satisfies $0\le b\le a\le1$. Condition on $X_3,\ldots,X_n$; the ordering $\pi$ remains uniformly random. Let $m$ be the number of indices $j\ge3$ with $X_j\ne0$. We repeatedly use the following fact: if exactly one index has label $2$ and $\ell$ indices have nonzero labels, then $E$ fails if and only if the index with label $2$ comes first among these $\ell$ indices, which has probability $1/\ell$.

If at least two indices $j\ge3$ have label $2$, then $E$ occurs.

If no index $j\ge3$ has label $2$, then $E$ occurs with probability one if $X_1=X_2=2$; with probability $(m+1)/(m+2)$ if one of $X_1,X_2$ equals $2$ and the other equals $1$; with probability $m/(m+1)$ if one of them equals $2$ and the other equals $0$; and with probability zero otherwise. Hence the conditional probability of $E$ is
\begin{align*}
&b_1b_2+\frac{m+1}{m+2}\bigl(b_1(a_2-b_2)+b_2(a_1-b_1)\bigr)+\frac{m}{m+1}\bigl(b_1(1-a_2)+b_2(1-a_1)\bigr)\\
&\qquad=\frac{m}{m+1}\,(b_1+b_2)+\frac{a_1b_2+a_2b_1}{(m+1)(m+2)}-\frac{m}{m+2}\,b_1b_2.
\end{align*}
By the arithmetic--geometric mean inequality, $b_1+b_2\ge2b$ and $a_1b_2+a_2b_1\ge2ab$, with equality after the replacement, while $b_1b_2=b^2$ is unchanged. Since the coefficients of the first two terms are nonnegative, the conditional probability of $E$ does not increase.

If exactly one index $j\ge3$ has label $2$, then $m\ge1$, and $E$ fails if and only if $X_1,X_2\in\{0,1\}$ and this index comes first among all indices with nonzero labels. Writing $N:=\ind[X_1=1]+\ind[X_2=1]$ and using $1/\ell=\int_0^1u^{\ell-1}\dd u$, the conditional failure probability equals
\[
\E\Bigl[\frac{\ind[X_1\ne2,\,X_2\ne2]}{m+N}\Bigr]
=\int_0^1u^{m-1}\prod_{i=1}^2\bigl(1-a_i+u(a_i-b_i)\bigr)\dd u.
\]
For each $u\in[0,1]$ and $i\in\{1,2\}$, the three nonnegative numbers $1-(1-u)a_i-ub_i$, $(1-u)a_i$, and $ub_i$ sum to one. Applying the Cauchy--Schwarz inequality to their square roots gives
\[
\sqrt{\prod_{i=1}^2\bigl(1-(1-u)a_i-ub_i\bigr)}+(1-u)a+ub\le1,
\]
that is, $\prod_{i=1}^2\bigl(1-a_i+u(a_i-b_i)\bigr)\le\bigl(1-a+u(a-b)\bigr)^2$. Hence the integrand, and thus the conditional failure probability, does not decrease under the replacement. Averaging over $X_3,\ldots,X_n$ shows that the replacement does not increase $\P[E]$.

To conclude, note that the set of parameter vectors with the given products is nonempty and compact, and $\P[E]$ is a polynomial in the parameters. Among the minimizers of $\P[E]$ over this set, choose one minimizing $\sum_i(a_i+b_i)$. If two of its parameter pairs differ, replacing both by their geometric means yields another minimizer by the pairwise step, and by the strict arithmetic--geometric mean inequality it strictly decreases $\sum_i(a_i+b_i)$, a contradiction. Hence all pairs coincide, and the product constraints force $a_i=A^{1/n}$ and $b_i=B^{1/n}$.

\end{proof}

\subsection[Proof of the Main Result]{Proof of Theorem~\ref{thm:last}}
\label{sec:last-proof}

\begin{proof}
Throughout the proof, $\P_{\dists}$ denotes probability when $\vals\sim\dists$ and the priority order $\pi$ is uniformly random and independent of $\vals$.
Fix $\dists=\F_1\times\cdots\times\F_n\in\Pi_n(\G)$, and let $s_i:=\bar\varphi(v_i)$ denote bidder $i$'s score.
Because $\bar\varphi$ is nondecreasing,
$\min_i s_i=\bar\varphi(v_{(n)})$, which has the same distribution under $\dists$ and under $\F^n$ because $v_{(n)}\sim\G$ under both. Independence of the values implies that, for every score threshold $c$,
\begin{equation}\label{eq:last-score-products}
\begin{aligned}
\prod_{i=1}^n\P_{v\sim\F_i}[\bar\varphi(v)\ge c]
&=\bigl(\P_{v\sim\F}[\bar\varphi(v)\ge c]\bigr)^n,\\
\prod_{i=1}^n\P_{v\sim\F_i}[\bar\varphi(v)>c]
&=\bigl(\P_{v\sim\F}[\bar\varphi(v)>c]\bigr)^n.
\end{aligned}
\end{equation}
Both strict and weak inequalities are needed because ironing may create
ties even when the value distributions are continuous.

For $t\ge0$, let $E_t$ be the event that the item is sold and its winner would not win if her bid were replaced by $t$, with the other bids and $\pi$ unchanged. For fixed $(\vals,\pi)$, the winner's allocation is nondecreasing in her own bid, so $E_t$ holds for every $t$ below her critical bid, which equals her payment, and fails for every $t$ above it. Hence her payment equals $\int_0^\infty\ind[E_t]\dd t$, and Tonelli's theorem gives
\begin{equation}\label{eq:last-critical-integral}
\Rev(\mech,\dists)
=\int_0^\infty\P_{\dists}(E_t)\dd t.
\end{equation}
We show that each integrand is minimized at $\F^n$.

If $\bar\varphi(t)<0$, then a bid of $t$ never wins, so $E_t$ is simply the event of a sale, that is, the event that some score is nonnegative.
Write $a_i=\P_{v\sim\F_i}[\bar\varphi(v)\ge0]$ and
$a=\P_{v\sim\F}[\bar\varphi(v)\ge0]$.
Equation~\eqref{eq:last-score-products} gives $\prod_i a_i=a^n$.
By the arithmetic--geometric mean inequality,
\[
\Bigl(\prod_i(1-a_i)\Bigr)^{1/n}+\Bigl(\prod_ia_i\Bigr)^{1/n}\le\frac1n\sum_i(1-a_i)+\frac1n\sum_ia_i=1,
\]
so $\prod_i(1-a_i)\le(1-a)^n$, and hence
\[
\P_{\dists}(E_t)
=1-\prod_i(1-a_i)
\ge1-(1-a)^n
=\P_{\F^n}(E_t).
\]

If $c:=\bar\varphi(t)\ge0$, assign to bidder $i$ the label $X_i=0$, $1$, or $2$
according as her score is below, equal to, or above $c$.
Then $E_t$ is precisely the event $E$ in Lemma~\ref{lem:two-threshold} for the order $\pi$.
Indeed, with at least two scores above $c$, the winner loses after bidding
$t$. With exactly one, the bidder with that score wins and, after bidding $t$, loses precisely when another bidder with
score $c$ has higher priority. With none, changing a winning bid to $t$
cannot make the winner lose.
The parameters of the labels are
\[
a_i=\P_{v\sim\F_i}[\bar\varphi(v)\ge c],
\qquad
b_i=\P_{v\sim\F_i}[\bar\varphi(v)>c].
\]
By~\eqref{eq:last-score-products}, their products equal $a^n$ and $b^n$, where $a$ and $b$ are the corresponding probabilities under $\F$; these are also the parameters of every label under $\F^n$. Lemma~\ref{lem:two-threshold}
therefore yields
$\P_{\dists}(E_t)\ge\P_{\F^n}(E_t)$ in this
case as well.

Integrating in~\eqref{eq:last-critical-integral} establishes
$\Rev(\mech,\dists)\ge\Rev(\mech,\F^n)$ for every
$\dists\in\Pi_n(\G)$.
Since $\F^n$ is feasible and $\mech$ is optimal for $\F^n$,
\[
\Rev(\mech,\F^n)
=\min_{\dists\in\Pi_n(\G)}\Rev(\mech,\dists)
\le\R[n](\G)
\le\OPT(\F^n)
=\Rev(\mech,\F^n).
\]
All inequalities are therefore equalities.
\end{proof}

\subsection{Discussion}
\label{sec:last-discussion}

\paragraph{Unknown Number of Bidders.}
For fixed $\G$, the robust revenue guarantee is nondecreasing in the
number of bidders. For $m\ge1$, let $\F^{(m)}:=1-(1-\G)^{1/m}$ denote the consistent i.i.d.\ marginal and $\mathcal{R}^{(m)}(\G)$ the robust benchmark when there are $m$ bidders and the distribution $\G$ of the smallest value is known. Since $(1-\G(v))^{1/(n+1)}\ge(1-\G(v))^{1/n}$ for every $v$, we have $\F^{(n+1)}\succeq\F^{(n)}$.
Keep the $n$-bidder auction's score function fixed, and run the auction with $n+1$ bidders and a uniformly random priority order. For each $t$, the event $E_t$ in the proof of Theorem~\ref{thm:last} is preserved when a bid increases or a bidder is added:
a sale remains a sale, labels can only increase, and the relative order of the existing bidders is unchanged. Couple the first $n$ bids so that
those drawn from $\F^{(n+1)}$ are no smaller than those drawn from $\F^{(n)}$, and insert the extra
bidder into a uniformly random position in the priority order.
By~\eqref{eq:last-critical-integral}, the resulting auction earns at least
$\OPT\bigl((\F^{(n)})^n\bigr)$ on $(\F^{(n+1)})^{n+1}$.
Optimizing for the enlarged market can only improve revenue, so
Theorem~\ref{thm:last} gives $\mathcal{R}^{(n)}(\G)\le\mathcal{R}^{(n+1)}(\G)$.
Thus a lower bound on participation provides a conservative revenue
benchmark.

\paragraph{Comparison with Known First Order Statistic.}
Both extreme-information settings admit a saddle-point characterization.
When the distribution of the largest valuation is known, an optimal
posted price earns the same revenue under every consistent distribution;
a distribution with one active bidder certifies optimality.
When the distribution of the smallest valuation is known, revenues need
not be identical across the ambiguity set. Instead, the consistent
i.i.d.\ distribution is least favorable for the Myerson auction with
uniform random priority.

\section[Second Order Statistic: k=2]{Second Order Statistic: $k=2$}

In practice, when the seller does not know the buyers' valuation distributions, a commonly adopted mechanism is a second-price auction or an ascending auction. These mechanisms are attractive because of their simplicity, transparency, and the existence of dominant strategies for bidders. Importantly, after running a second-price auction repeatedly, the seller can form reliable estimates of the distribution $\G$ of the final transaction price, i.e., the distribution of the second order statistic of bidders’ values.

In this section, we study the setting in which the distribution of an intermediate $k$-th order statistic is known and aim to design a robust mechanism. Readers may keep the case $k=2$ in mind as a motivating example, although our analysis applies to general $k$. Beyond its theoretical interest, we also discuss applications of general $k$-th order statistics at the end of the section.
This setting is significantly more challenging than the cases $k=1$ and $k=n$ because, as we show below, unless $\G$ is a point mass, no distribution in the ambiguity set minimizes all order statistics simultaneously.

\paragraph{Non-Existence of a Simultaneously Minimal Distribution.}
We say that $\dists^*\in\Pi_k(\G)$ is \emph{simultaneously minimal} if
\begin{equation}\label{eq:simultaneous}
\OS[i](\dists^*)\preceq\OS[i](\dists)\qquad\text{for all } i\in[n] \text{ and all } \dists\in\Pi_k(\G).
\end{equation}
Such a distribution would be least favorable for every mechanism whose revenue is a sum of nondecreasing functions of individual order statistics. When $k=1$, the distribution $\G \times \delta_0 \times \cdots \times \delta_0$ is simultaneously minimal, since under every $\dists\in\Pi_1(\G)$ the largest value has distribution $\G$ and all other values are nonnegative. When $k=n$, Lemma~\ref{lem:general-order-statistic} below shows that $\Fiid[n](\G)^n$ is simultaneously minimal.

The key technical ingredient of this section is the following observation.
For each index $i \in [n]$, there does exist a
distribution that is minimal with respect to the $i$-th order statistic.
Crucially, however, these minimal distributions depend on whether $i$ is smaller or larger than $k$. Proposition~\ref{prop:nonexistence} below shows that, when $\G$ is not a point mass, they cannot be reconciled unless $k=1$ or $k=n$, in which case either $i \ge k$ for all $i\in[n]$ or $i \le k$ for all $i \in [n]$.

\begin{lemma}
\label{lem:general-order-statistic}
For every $\dists \in \Pi_k(\G)$, the following statements hold.
\begin{itemize}
    \item If $i < k$, then
    $\OS[i](\dists) \succeq \OS[i]\left(\Fiid[k](\G)^n\right)$.
    Moreover, equality is attained when
    $\dists = \Fiid[k](\G)^n$.
    \item If $i > k$, then
    $\OS[i](\dists) \succeq \delta_0$.
    Moreover, equality is attained when $\dists = \F_1 \times \cdots \mathbin{\times} \F_k \times \delta_0 \times \cdots \times \delta_0$, provided that $\OS(\dists) = \G$.
\end{itemize}
\end{lemma}

\begin{proposition}\label{prop:nonexistence}
Let $2\le k\le n-1$, and let $\G$ be a distribution on $[0,\infty)$ that is not a point mass. Then $\Pi_k(\G)$ contains no simultaneously minimal distribution.
\end{proposition}

The proof is given in Section~\ref{sec:k-statistic}. The assumption on $\G$ cannot be dropped: if $\G=\delta_a$ for some $a\ge0$, then $\delta_a^{k}\times\delta_0^{n-k}$ is simultaneously minimal in $\Pi_k(\G)$, because every $\dists\in\Pi_k(\G)$ has $v_{(k)}=a$ almost surely.

The proof of Lemma~\ref{lem:general-order-statistic} is technical and is deferred to Section~\ref{sec:k-statistic}. Below, we provide a simpler proof tailored to the special case $k=2$ and $i=1$, which suffices for the analysis of the setting in which the distribution of the second order statistic is known.
\begin{lemma}
\label{lem:second_first}
For every $\dists \in \Pi_2(\G)$, we have $\OS[1](\dists) \succeq \OS[1](\Fiid[2](\G)^n)$.
\end{lemma}
\begin{proof}
It suffices to show that for every $v \in \mathbb{R}$, the probability $\Pr[\vals \sim \dists]{v_{(1)} \le v}$ is maximized when $\dists = \Fiid[2](\G)^n$.

Fix $v$ and denote $x_i := \F_i(v)=\Pr{v_i \le v}$ for each $i\in[n]$. Since $v_{(2)} \sim \G$, we have
\begin{align*}
\G(v) & = \prod_{i=1}^n\Pr{v_i \le v} +  \sum_{i=1}^n \Big( \Pr{v_i > v} \cdot \prod_{j\neq i}\Pr{v_j \le v} \Big) = \prod_{i=1}^n x_i +  \sum_{i=1}^n \Big( (1-x_i)\cdot \prod_{j\neq i}x_j \Big) \\
& = \sum_{i=1}^n \prod_{j \ne i} x_j - (n-1) \prod_{i=1}^n x_i
\end{align*}
Moreover, the distribution of the largest order statistic satisfies
\[
\OS[1](\dists)(v) = \Pr[\vals \sim \dists]{v_{(1)} \le v} = \prod_{i=1}^{n} \F_i(v) = \prod_{i=1}^{n} x_i.
\]
Therefore, it suffices to show that the following optimization problem achieves its maximum when $x_1=\cdots=x_n$:
\begin{align}
\begin{aligned}
    \max_{x_1,\cdots,x_n\in [0,1]} & \prod_{i=1}^n x_i \\
    \text{subject to: } & \sum_{i=1}^n \prod_{j \ne i} x_j - (n-1) \prod_{i=1}^n x_i = \G(v)
\end{aligned}
\label{eq:second}
\end{align}
Let $\bar{x}$ denote the geometric mean of $\{x_i\}_{i=1}^n$, i.e., $\bar{x} := \left( \prod_{i=1}^n x_i \right)^{1/n}$.
By the AM-GM inequality,
\[
\G(v) = \sum_{i=1}^n \prod_{j \ne i} x_j - (n-1) \prod_{i=1}^n x_i \ge n \cdot \bar{x}^{n-1} - (n-1) \bar{x}^n~.
\]
Define $h_n(x) := n x^{n-1} - (n-1)x^{n}$. The function $h_n$ is strictly increasing on $[0,1]$, with $h_n(0)=0$ and $h_n(1)=1$.
Since $0 \le \G(v) \le 1$, there exists a unique $x^* \in [0,1]$ such that
$h_n(x^*) = \G(v)$. By the monotonicity of $h_n$, the above inequality implies $\bar{x} \le x^*$.

Finally, note that setting $x_1=\cdots=x_n=x^*$ is feasible for \eqref{eq:second} and yields $\prod_{i=1}^n x_i = (x^*)^n \ge \prod_{i=1}^n x_i$ for any feasible $(x_1,\dots,x_n)$. This proves that the maximum is attained when all $x_i$ are equal, and thus completes the proof.
\end{proof}

\subsection{Robustness of Second Price Auction with a Reserve}
\label{sec:second}
The contrast between the first and last order cases provides useful insights. When first order information is available, robust optimality favors mechanisms that effectively ignore competition and screen a single high type. When last order information is available, robust optimality instead favors mechanisms that fully exploit competition among symmetric bidders. Guided by this intuition, when the seller observes the second order statistic, it is natural to focus on mechanisms that extract revenue from the interaction between the highest and second-highest bidders.
For the widely used second-price auction with a reserve price, we show that the worst-case revenue is attained by the i.i.d.\ distribution $\Fiid[2](\G)^n$. This is because the revenue of a second-price auction with reserve depends only on the first and second order statistics of the valuation profile $\vals$.
Consequently, the second-price auction with reserve admits a sharp and tractable worst-case characterization, allowing the seller to secure a robust lower bound on revenue across all independent distributions consistent with the observed winning price distribution $\G$.

\begin{theorem}
\label{thm:second}
For every $k \ge 2$ and price $p \ge 0$,
\[
\min_{\dists \in \Pi_k(\G)} \Rev(\spa(p), \dists) = \Rev(\spa(p), \Fiid(\G)^n)~.
\]
Moreover, if $\Fiid(\G)$ is regular above the reserve $p^* = \argmax_v v \cdot (1-\Fiid(\G)(v^-))$, we have
\[
\min_{\dists \in \Pi_k(\G)} \Rev(\spa(p^*), \dists) = \Rev(\spa(p^*), \Fiid(\G)^n) = \R(\G)~.
\]
\end{theorem}
\begin{proof}
By equation~\eqref{eqn:revenue_spa}, the seller's revenue under $\spa(p)$ is
\begin{align*}
\Rev(\spa(p), \dists) & = p \cdot \Pr[\vals \sim \dists]{v_{(1)} \ge p > v_{(2)}} + \Ex[\vals \sim \dists]{v_{(2)} \cdot \ind[v_{(2)} \ge p]} \\
& = p \cdot \Pr[\vals \sim \dists]{v_{(1)} \ge p} + \Ex[\vals \sim \dists]{(v_{(2)}-p) \cdot \ind[v_{(2)} \ge p]} \\
& = p \cdot \Pr[\vals \sim \dists]{v_{(1)} \ge p} + \int_p^\infty \Pr[\vals \sim \dists]{v_{(2)} \ge v}  \dd v
\end{align*}

Applying Lemma~\ref{lem:general-order-statistic} to $i=1,2$ (when $k=2$, the distribution of $v_{(2)}$ is $\G$ under every $\dists\in\Pi_2(\G)$), we conclude
that the terms $\Pr{v_{(1)} \ge p}$ and $\Pr{v_{(2)} \ge v}$, for every $v\ge p$, attain their
minimum values at the distribution $\Fiid(\G)^n$. Therefore,
\begin{align*}
\Rev(\spa(p), \dists) & \ge p \cdot \Pr[\vals \sim \Fiid(\G)^n]{v_{(1)} \ge p} + \int_p^\infty \Pr[\vals \sim \Fiid(\G)^n]{v_{(2)} \ge v}  \dd v  = \Rev(\spa(p), \Fiid(\G)^n)~.
\end{align*}
Together with the fact that $\Fiid(\G)^n \in \Pi_k(\G)$, this establishes the first part of the statement.

We remark that when $k=2$, the second component of the revenue,
\[
\int_p^\infty \Pr[\vals \sim \dists]{v_{(2)} \ge v} \dd v = \int_p^\infty 1-\G(v) \dd v,
\]
is constant across all distributions $\dists \in \Pi_2(\G)$.
Consequently, the weaker Lemma~\ref{lem:second_first} suffices for the above analysis.

The second part of the statement comes from the following observation:
\[
\R(\G) = \sup_\mech \inf_{\dists \in \Pi_k(\G)} \Rev(\mech,\dists) \le \sup_\mech \Rev(\mech,\Fiid(\G)^n) = \Rev(\spa(p^*), \Fiid(\G)^n)~,
\]
where the last equality holds by Fact~\ref{fact:spa}. Conversely, the first part gives
\[
\R(\G)\ge\min_{\dists\in\Pi_k(\G)}\Rev(\spa(p^*),\dists)=\Rev(\spa(p^*),\Fiid(\G)^n).
\]

\end{proof}

\paragraph{Discussion.}
Our theorem holds for general $k \ge 2$. Below, we make a few remarks in the context of $k=2$, i.e., when the distribution of the second order statistic is known.

\paragraph{Unknown Number of Bidders.}
Fixing the distribution $\G$ of the second order statistic, we observe that as the number of bidders increases, the worst-case revenue achieved by $\spa(p)$ for any fixed reserve price $p$ weakly decreases. Indeed, given any joint distribution of valuations among $n$ bidders, one can construct a feasible joint distribution for $n+1$ bidders by introducing an additional bidder whose valuation is identically zero. This construction leaves the distribution of the second order statistic unchanged.
Consequently, any second-order-statistic distribution that is feasible with $n$ bidders remains feasible in the robust auction problem with $n+1$ bidders. Since the robust benchmark for $n+1$ bidders minimizes revenue over a larger ambiguity set, the worst-case revenue attained by $\spa(p)$ with $n+1$ bidders is weakly lower than that with $n$ bidders.

This suggests that when the number of bidders is unknown, the optimal reserve price $p$ in the limiting regime as $n \to \infty$ obtains a lower bound for the robust revenue for all $n$. In Proposition~\ref{prop: unknown n}, we provide, for every reserve price, a revenue guarantee that holds uniformly for all $n$.
\begin{proposition}
For any reserve price $p\ge0$, the second-price auction with a reserve price $p$ guarantees the following worst-case expected revenue for every $n$ and every distribution $\dists \in \Pi_2(\G)$:
\[
\Rev(\spa(p), \dists) \ge p \cdot (1-z^*(p)) + \int_p^\infty \bigl(1-\G(v)\bigr) \dd v~,
\]
where $z^*(p)$ is the unique solution in $[0,1]$ of $z(1-\ln z)=\G(p^-)$, with the convention $0\cdot\ln 0=0$.
\label{prop: unknown n}
\end{proposition}
\begin{proof}
According to the proof of Theorem~\ref{thm:second}, for any reserve price $p$,
\[
\Rev(\spa(p), \dists)= p \cdot \Pr[\vals \sim \dists]{v_{(1)} \ge p} + \int_p^\infty \bigl(1-\G(v)\bigr) \dd v.
\]
Thus, it suffices to provide a lower bound on $\P\off{v_{(1)}\ge p}$ that holds for every number of bidders $n$.
We prove that for every $n$,
\begin{equation}
\label{eqn:unknown_second}
    \P[v_{(1)} \ge p] \ge 1-z^*(p)~.
\end{equation}

Fix $n$, and let $\F:=\Fiid[2](\G)$ and $x:=\F(p^-)$. By Lemma~\ref{lem:second_first}, $\Pr[\vals\sim\dists]{v_{(1)}\ge p}\ge\Pr[\vals\sim\F^n]{v_{(1)}\ge p}=1-x^n$. Letting $v\uparrow p$ in the identity $\F(v)^n+n\F(v)^{n-1}(1-\F(v))=\G(v)$ gives
\[
x^n + n x^{n-1}(1-x) = \G(p^-)~.
\]
It therefore suffices to show that $z:=x^n \le z^*(p)$. This is immediate if $x=0$. Otherwise,
\begin{align*}
\G(p^-) & = x^n + n x^{n-1}(1-x) \\
& = z + z \cdot n(1-x)/x \\
& \ge z + z \cdot (- n\ln x) = z - z \ln z~,
\end{align*}
where the inequality follows from the bound $\frac{1-x}{x} \ge -\ln x$ for $x\in(0,1]$.
The inequality~\eqref{eqn:unknown_second} then follows from the fact that the function $z \mapsto z-z\ln z$ is increasing on $[0,1]$. Moreover, equality is attained in the limit as $n \to \infty$. Indeed, $z_n:=\Fiid[2](\G)(p^-)^n$, computed for $n$ bidders, solves $\theta_n(z_n)=\G(p^-)$, where $\theta_n(z):=z+n\,z\,(z^{-1/n}-1)$ is continuous and increasing on $[0,1]$. As $n\to\infty$, $\theta_n$ decreases pointwise to the continuous function $\theta(z):=z(1-\ln z)$, hence uniformly on $[0,1]$ by Dini's theorem, so $z_n\to z^*(p)$.
To summarize, when the number of bidders is unknown, the second-price auction with a reserve achieves a uniform revenue guarantee for all $n$.
\end{proof}

\paragraph{Regularity vs.\ Regularity above the Reserve.}
Most of the literature states the optimality of the second-price auction with a reserve price under the assumption that the underlying distribution is regular.
In contrast, we deliberately adopt the more precise notion of \emph{regularity above the reserve}. Indeed, the most relevant applications of our result arise when $\Fiid[2](\G)$ is regular above the reserve but not regular.

A classical result of \cite{BulowK1996} shows that when bidders are i.i.d.\ and regular, the revenue achieved by the second-price auction without a reserve under $n+1$ bidders exceeds that obtained under $n$ bidders with the optimal reserve.
Conceptually, this result suggests that increasing competition is often more effective than optimizing the auction format.

Viewed through the lens of our framework, this result implies that when $\Fiid[2](\G)$ is regular, although optimizing a reserve price is robustly optimal, the resulting guaranteed revenue gain over the second-price auction without a reserve (from which the data are collected), i.e., the gain under the consistent i.i.d.\ distribution, is negligible when the number of buyers is large---at most a $1/(n-1)$ fraction. Formally, let $\F:=\Fiid[2](\G)$. Applying the result of \cite{BulowK1996} to $n-1$ bidders gives $\Rev(\spa(0),\F^n)\ge\opt(\F^{n-1})$. Moreover, $\opt(\F^{n-1})\ge\frac{n-1}{n}\opt(\F^n)$, because for i.i.d.\ bidders $\opt(\F^m)$ is the expected maximum of the nonnegative parts of $m$ i.i.d.\ (ironed) virtual values, and a uniformly random subset of $n-1$ of $n$ exchangeable random variables contains their maximum with probability $(n-1)/n$. Hence, by Fact~\ref{fact:spa},
\begin{align*}
\Rev(\spa(p^*), \F^n) - \Rev(\spa(0), \F^n)
&= \opt(\F^n) - \Rev(\spa(0), \F^n)
\le \frac{1}{n}\, \opt(\F^n) \\
&\le \frac{1}{n-1} \cdot \Rev(\spa(0), \F^n)
= \frac{1}{n-1} \cdot \Ex[v \sim \G]{v}.
\end{align*}
The factor $1/(n-1)$ cannot be replaced by $1/n$: if the revenue curve of $\F$ is piecewise linear, increasing from $(0,0)$ to $(a,a)$ and decreasing from $(a,a)$ to $(1,0)$, then the revenue gain equals $a(1-a)^{n-1}$ and $\Rev(\spa(0),\F^n)=1-(1-a)^{n-1}$, so their ratio, multiplied by $n-1$, tends to one as $a\to0$. (This distribution has an atom at $1$ and is regular in the sense that its revenue curve is concave; smoothing the atom yields distributions with densities and the same conclusion.)

In contrast, when $\Fiid[2](\G)$ is irregular but regular above the reserve, our results point to the opposite conclusion. In this case, if the seller has collected sufficient data by running a second-price auction without a reserve, optimizing the reserve price can lead to a \emph{constant-factor} revenue improvement, even when the number of buyers is large or unknown.
This phenomenon is driven by the fact that the value distribution of an infrequent regular buyer is no longer regular. More formally, a mixture of a regular distribution with a point mass $\delta_0$ is not regular, which serves as a key motivation for our results.
We provide a few examples to illustrate the potential revenue improvement.
\paragraph{Example 1: Bernoulli Distribution.} Consider the extreme case in which $\G$ is a Bernoulli distribution, taking values $0$ and $1$ with equal probability. In this case, the induced distribution $\Fiid[2](\G)$ is also a Bernoulli distribution, taking values $0$ and $1$ with probabilities $q$ and $1-q$, respectively, where $q$ satisfies
\[
n q^{n-1} - (n-1) q^n = 0.5~.
\]
This distribution can be interpreted as a regular buyer whose valuation is deterministically $1$, but who appears only with probability $1-q$. Importantly, this distribution is regular above the reserve price but not regular.
Against this distribution, the second-price auction with a reserve price of $1$ is equivalent to a posted-price mechanism with price $1$. According to Proposition~\ref{prop: unknown n}, for every $n$, this yields a revenue of at least
\[
1 - z^* \approx 0.813, \qquad \text{where } z^*(1 - \ln z^*) = \G(1^-) \mathrel{=} 0.5~.
\]
This represents a substantial improvement over the expected revenue of $0.5$ obtained by running a second-price auction without a reserve.

\paragraph{Example 2: Uniform Distribution.}
Consider the case when $\G$ is the uniform distribution on $[0,1]$. We first prove that the induced distribution $\Fiid[2](\G)$ is not regular but is regular above the reserve for every $n\ge3$, and then establish the constant revenue improvement.
\begin{claim}
    For every $n \ge 3$, the distribution $\Fiid[2](U[0,1])$ is regular above the reserve, but not regular.
\end{claim}
For $n=2$, the distribution $\Fiid[2](U[0,1])(v)=1-\sqrt{1-v}$ is regular, with virtual value $\varphi(v)=3v-2$.
\begin{proof}
Consider an arbitrary $n \ge 3$. For notation simplicity, we use $\F$ to denote $\Fiid[2](U[0,1])$ and $f$ to denote the corresponding probability density function. By definition, we have
\[
\F(v)^n + n \F(v)^{n-1} (1-\F(v)) = \G(v) = v, \qquad \forall v \in [0,1]~.
\]
Taking derivatives on both sides yields
\[
n(n-1) \F(v)^{n-2}(1-\F(v)) f(v) = 1, \qquad \forall v \in [0,1]~.
\]
The virtual value function associated with $\F$ is therefore
\[
\varphi(v) = v - \frac{1-\F(v)}{f(v)} = n \F(v)^{n-1} - (n-1) \F(v)^{n} - n(n-1) \F(v)^{n-2} (1-\F(v))^2~.
\]
Since $\F$ is strictly increasing in $v$, the monotonicity of $\varphi(v)$ is equivalent to the monotonicity of the function
\[
\psi(y) := n y^{n-1} - (n-1)y^n -n(n-1)y^{n-2}(1-y)^2~.
\]
A direct computation shows that
\begin{align*}
\psi'(y) & = n(n-1)y^{n-2}(1-y) - n(n-1) ((n-2)y^{n-3}(1-y)^2 - 2y^{n-2}(1-y)) \\
& = n(n-1)y^{n-3}(1-y) \cdot ((n+1)y - (n-2))~.
\end{align*}
It follows that $\psi'(y) < 0$ for $y < \frac{n-2}{n+1}$ and $\psi'(y) > 0$ for $y > \frac{n-2}{n+1}$.
Equivalently, the virtual value function $\varphi(v)$ is decreasing for $v \le \F^{-1}\!\left(\frac{n-2}{n+1}\right)$ and increasing thereafter. Hence, $\F$ is not regular.

We now turn to regularity above the reserve. The monopoly reserve price $p^*$ is characterized by the condition $\varphi(p^*) = 0$, or equivalently by the unique solution $y^* = \F(p^*)$ satisfying
\[
\psi(y) = n(n-1)y^{n-2} \left( -\left(1+\frac{1}{n}\right) y^2 + \left(2+\frac{1}{n-1}\right) y - 1 \right) = 0~.
\]
Since the quadratic factor equals $-1$ at $y=0$ and $n\ge3$, we have $\psi(y)<0$ for all sufficiently small $y>0$; as $\psi$ is decreasing on $(0,\frac{n-2}{n+1}]$, it is negative on this whole interval. Moreover, $\psi$ is increasing on $[\frac{n-2}{n+1},1]$ and $\psi(1)=1>0$. Hence $\psi$ has a unique root $y^*\in(0,1)$, which satisfies $y^* > \frac{n-2}{n+1}$, and $\varphi<0$ below $p^*=\F^{-1}(y^*)$ while $\varphi>0$ above it. Because $\frac{d}{dv}\,v(1-\F(v))=-f(v)\varphi(v)$, the price $p^*$ is indeed the unique monopoly price. Therefore, $\varphi(v)$ is non-decreasing for all $v \ge p^*$, and consequently $\F$ is regular above the reserve.
\end{proof}
Our theorem then implies that the second-price auction with a reserve is robustly optimal for every $n\ge2$.
Moreover, the preceding discussion for an unknown number of buyers in \Cref{prop: unknown n} shows that by introducing a reserve price $p \in [0,1]$, the seller can guarantee an expected revenue of
\[
p (1-z^*(p)) + 0.5 (1-p)^2~, \quad \text{where } z^*(p)(1-\ln z^*(p)) = \G(p) = p~.
\]
Instead of optimizing over $p$, we can equivalently optimize over $z$ by substituting the expression for $p$ in terms of $z$, which yields the objective
\[
z(1-\ln z)(1-z) + 0.5 (1-z(1-\ln z))^2~.
\]
The maximum of this function is attained at $z^* \approx 0.198$, or equivalently $p^* \approx 0.519$, resulting in a revenue guarantee of approximately $0.532$. Once again, this represents a constant-factor revenue improvement over the second-price auction without a reserve for all $n$ uniformly.

\paragraph{More Examples by Numerical Illustrations.}
Regularity above reserve has many applications, and our analysis extends well beyond the above mentioned distribution classes of $\G$. In particular, we verify numerically, for representative members of widely used distribution families such as exponential, Beta, or normal distributions (the latter truncated to $[0,\infty)$), that the induced i.i.d.\ distribution $\Fiid[2](\G)$ is also regular above its optimal reserve. Since the computation is complicated and does not provide additional insights, we illustrate the regular-above-reserve property directly through the corresponding revenue curves in Figure~\ref{fig:regular_above_example}. Within this broad class of distributions, Theorem~\ref{thm:second} implies that the second-price auction with the optimal reserve is robustly optimal.
\begin{figure}[t]
\centering
\begin{subfigure}[t]{0.32\textwidth}
\includegraphics[width=\textwidth,keepaspectratio]{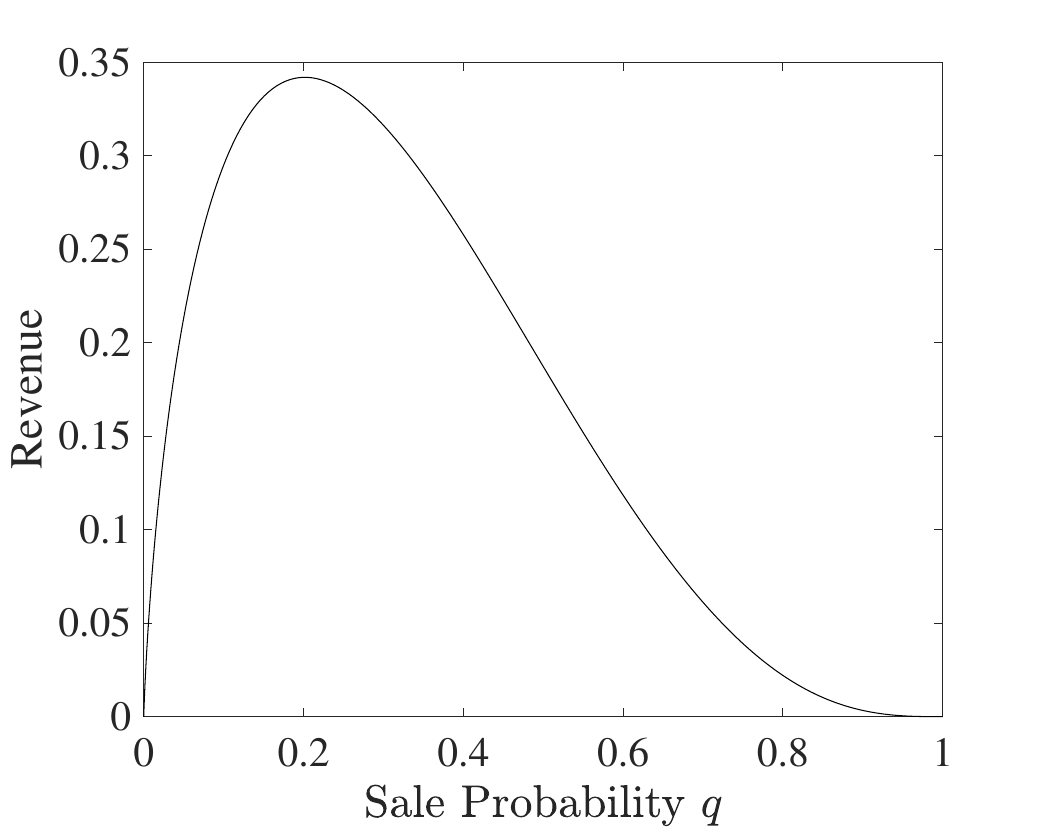}
\caption{$\G\sim \mathrm{exp}(1), n =4$}
\end{subfigure}
\begin{subfigure}[t]{0.32\textwidth}
\includegraphics[width=\textwidth,keepaspectratio]{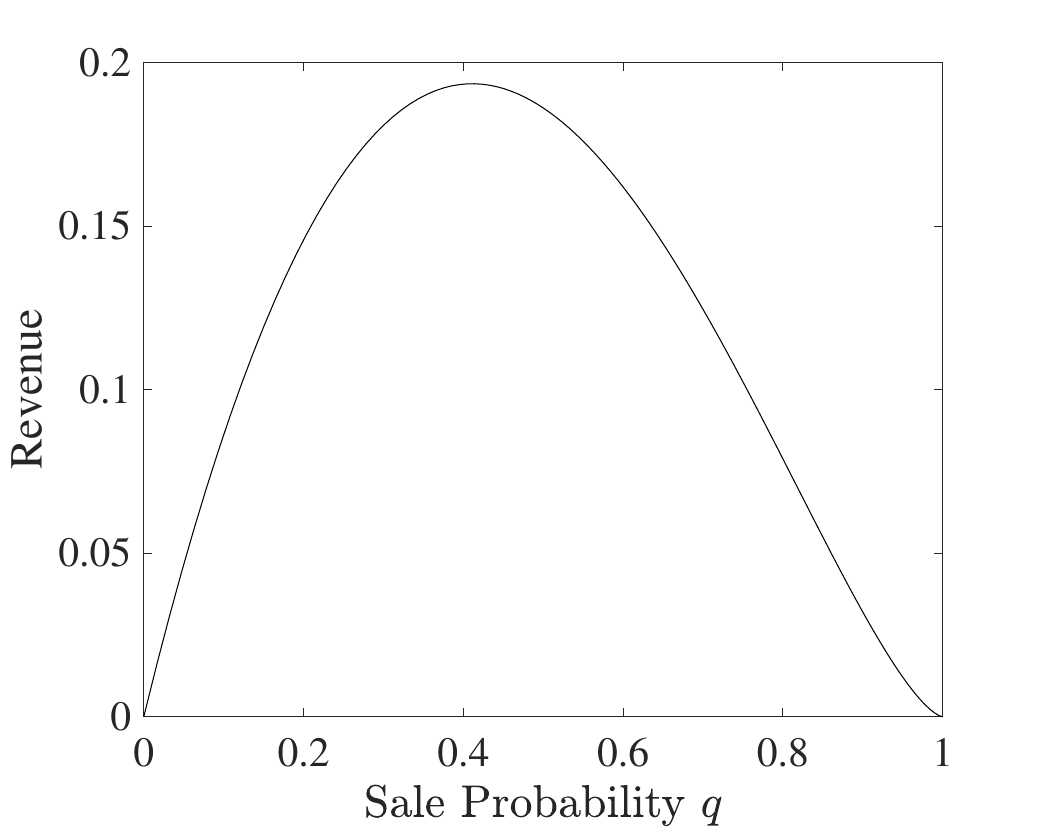}
\caption{$\G\sim\mathrm{Beta}(2,2), n =4$}

\end{subfigure}
\begin{subfigure}[t]{0.32\textwidth}
\includegraphics[width=\textwidth,keepaspectratio]{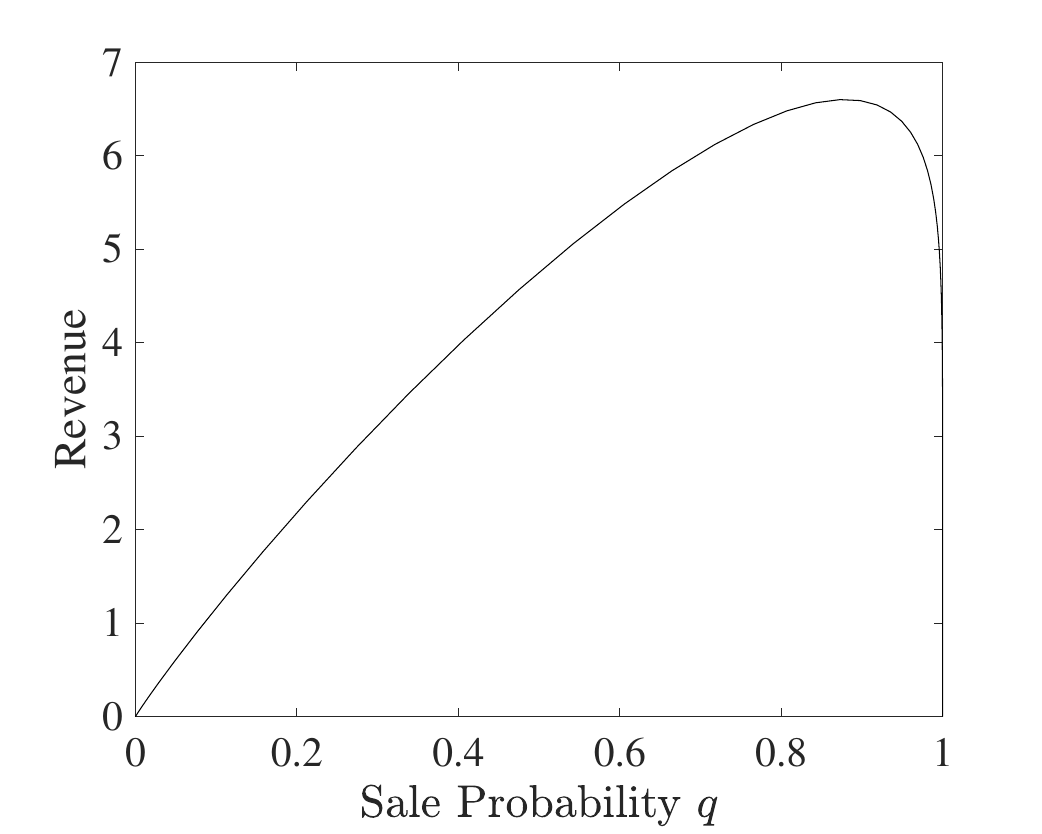}
\caption{$\G\sim \mathrm{Normal}(10,1), n =4$}
\end{subfigure}

\begin{subfigure}[t]{0.32\textwidth}
\includegraphics[width=\textwidth,keepaspectratio]{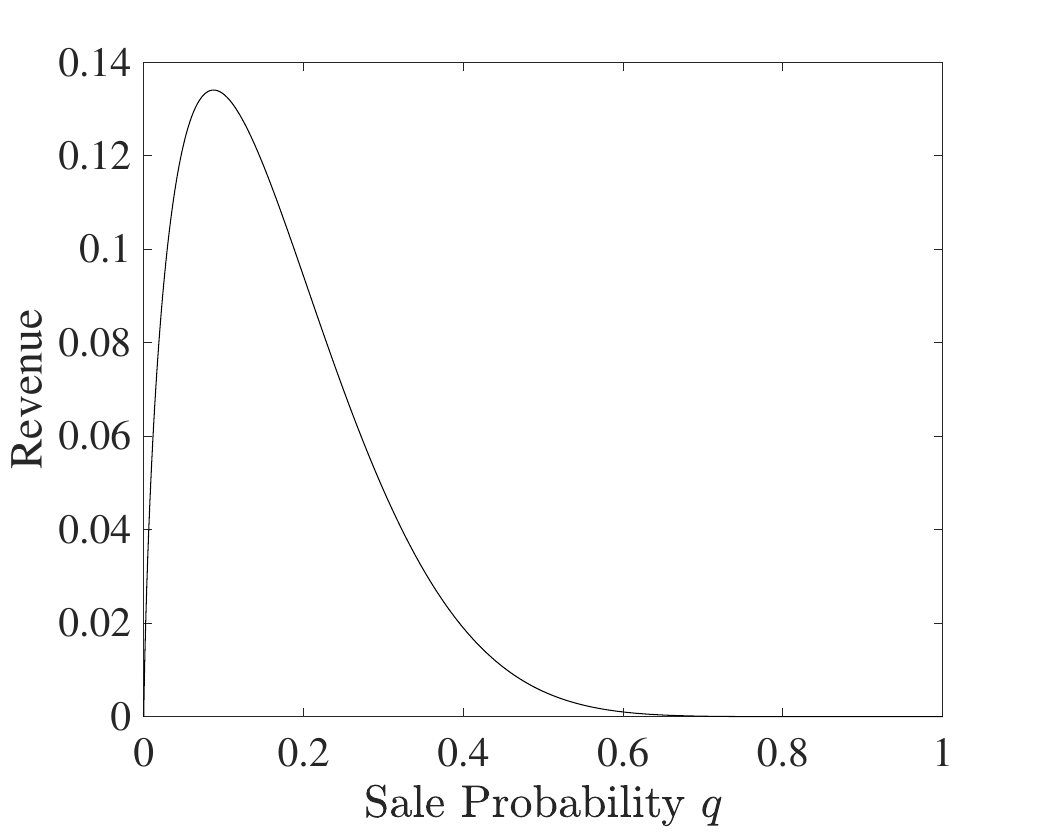}
\caption{$\G\sim \mathrm{exp}(1), n =10$}
\end{subfigure}
\begin{subfigure}[t]{0.32\textwidth}
\includegraphics[width=\textwidth,keepaspectratio]{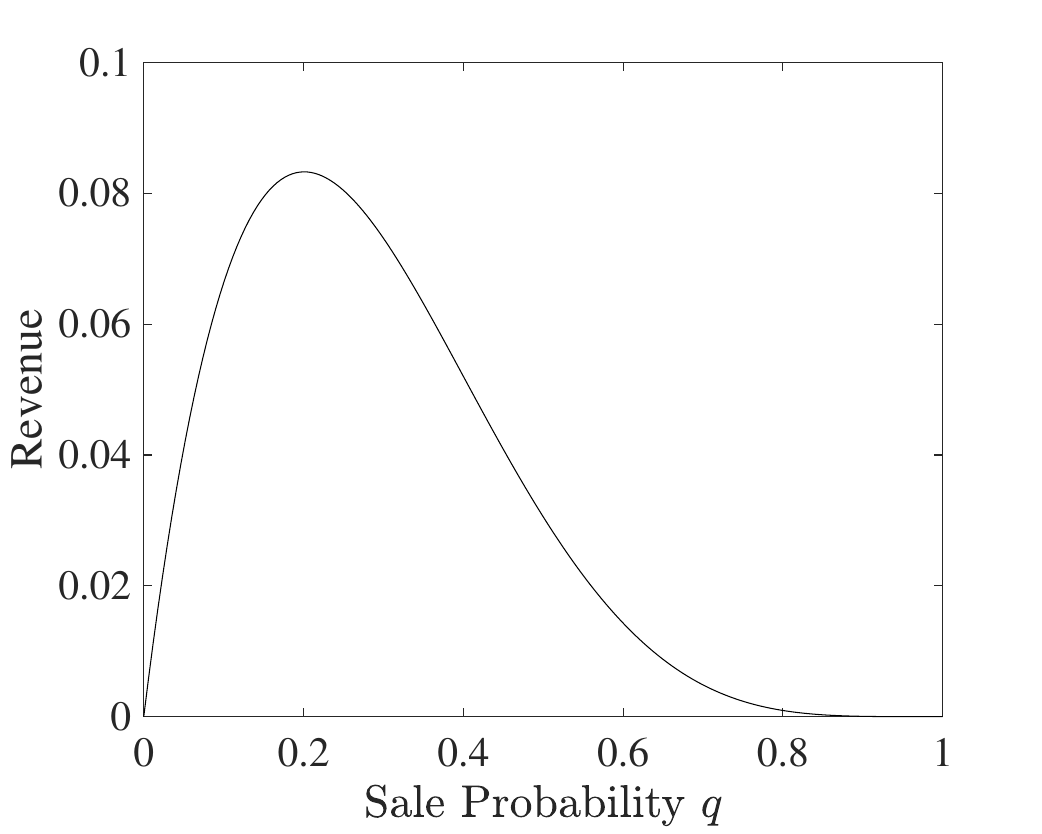}
\caption{$\G\sim \mathrm{Beta}(2,2), n =10$}
\end{subfigure}
\begin{subfigure}[t]{0.32\textwidth}
\includegraphics[width=\textwidth,keepaspectratio]{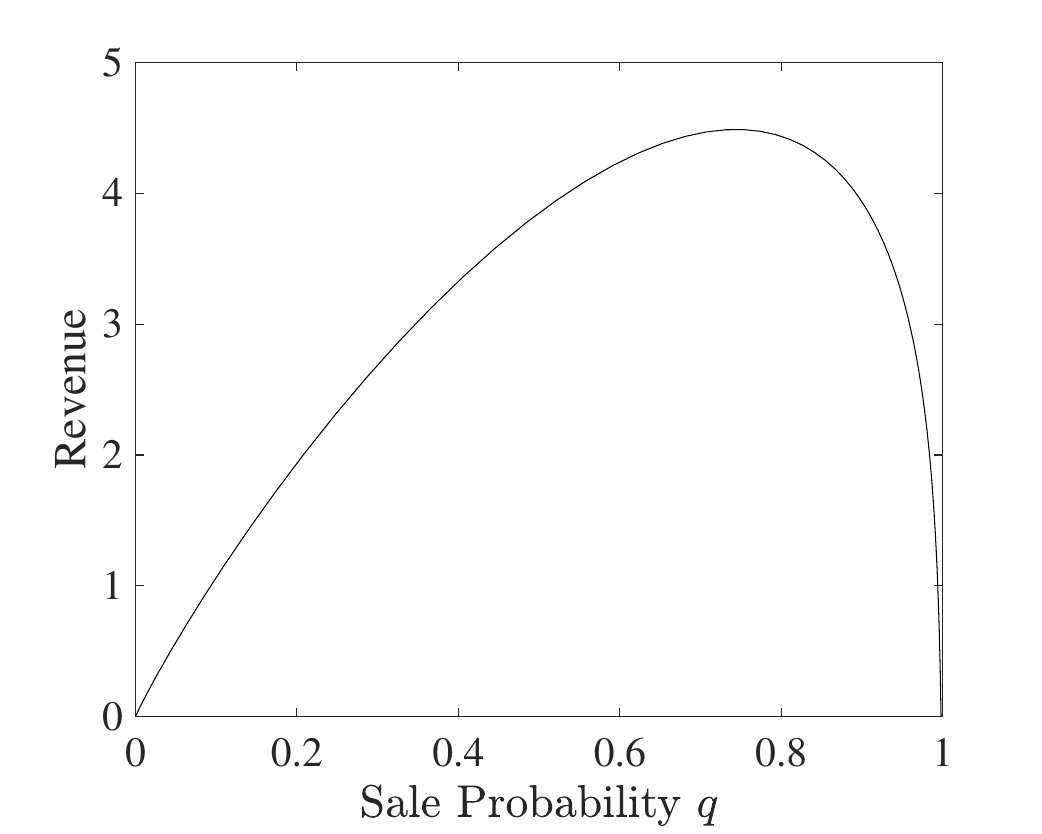}
\caption{$\G\sim \mathrm{Normal}(10,1), n =10$}
\end{subfigure}

\captionsetup{justification=centering}
\caption{Revenue Curves of $\Fiid[2](\G)$ when $\G$ is exponential, Beta or normal distribution}
\label{fig:regular_above_example}
\end{figure}

\subsection{Non-Robustness of Myerson's Auction}
\label{sec:example}

Next, we examine the robustness of Myerson’s auction. When $\Fiid(\G)$ is regular above the reserve, Myerson’s auction coincides with the second-price auction with a reserve price, and its robustness follows from our previous analysis. In this subsection, we focus on more subtle cases in which ironing is required.

At first glance, since Myerson’s auction allocates the item to the buyer with the highest ironed virtual value and charges the corresponding minimum winning bid, it may appear that the auction depends only on the first and second order statistics.
However, this intuition is misleading. The ironing procedure introduces subtle tie-breaking issues. In particular, within each ironed interval, ironing creates a flat region in the ironed virtual value function, which induces point masses
in the distribution of ironed virtual values. As a consequence, there is a positive probability that multiple buyers share the same ironed virtual value.

These ties must be resolved by a tie-breaking rule, and the resulting allocation and payments depend on the number of bidders whose bids tie at the second-highest level.
In this way, the revenue guarantee of Myerson’s auction implicitly depends on lower order statistics, rendering its robust analysis significantly more delicate.

In the remainder of this section, we construct a concrete counterexample showing that
\[
\min_{\dists \in \Pi_k(\G)} \Rev(\myerson(\Fiid(\G)^n), \dists) < \Rev(\myerson(\Fiid(\G)^n), \Fiid(\G)^n)~.
\]
In other words, identifying the unique i.i.d.\ distribution $\Fiid(\G)$ that is consistent with $\G$ and running Myerson’s auction designed for $\Fiid(\G)^n$ can be overly optimistic.

In fact, we establish a stronger statement in our example: no mechanism can guarantee the revenue achieved by Myerson’s auction on the consistent i.i.d.\ distribution $\Fiid(\G)$ across all distributions in $\Pi_k(\G)$.
Moreover, a variant of the same example admits the second-price auction with a reserve as the robustly optimal mechanism, whereas Myerson’s auction fails to be robustly optimal.

\paragraph{Example.}
We define two distributions on $[1,2]$.
Let $\F_{\mathrm{disc}}$ denote the two-point distribution supported on $\{1,2\}$, where
\[
v =
\begin{cases}
1, & \text{with probability } q,\\
2, & \text{with probability } 1-q.
\end{cases}
\]

Let $\F_{\mathrm{reg}}$ denote the distribution with cumulative distribution function
\[
\F_{\mathrm{reg}}(v) =
\begin{cases}
0, & v \in [0,1),\\[6pt]
\dfrac{v - 1}{\,v - 2 + 1/q\,}, & v \in [1,2),\\[10pt]
1, & v \in [2,\infty).
\end{cases}
\]
It has a density on $[1,2)$ and an atom of mass $1-q$ at $2$.
We focus on the regime where $q$ is sufficiently large:
\[
3q^2 - 2q^3 \ge \tfrac{3}{4} \iff q \ge q_0 \approx 0.674,
\]
where $q_0$ is the unique root of $3q^2-2q^3=\frac34$ in $(\frac12,1)$. We also assume $q<1$, because the construction degenerates at $q=1$, where $\F_{\mathrm{disc}}=\delta_1$.
This is a technical assumption, whose use shall be clear from the analysis.
The revenue curve of $\F_{\mathrm{reg}}$ is the concave envelope of the revenue curve of $\F_{\mathrm{disc}}$. Equivalently, $\F_{\mathrm{reg}}$ is the ironed version of $\F_{\mathrm{disc}}$.
The ironed virtual value function $\bar{\varphi}_{\mathrm{disc}}$ of $\F_{\mathrm{disc}}$ coincides with the virtual value function $\varphi_{\mathrm{reg}}$ of $\F_{\mathrm{reg}}$ and is equal to the derivative of the revenue curve:
\[
\bar{\varphi}_{\mathrm{disc}}(2) = \varphi_{\mathrm{reg}}(2) = 2, \qquad
\bar{\varphi}_{\mathrm{disc}}(v) = \varphi_{\mathrm{reg}}(v) = 2 - 1/q, \quad \forall v \in [1,2)~.
\]
In particular, $\F_{\mathrm{reg}}$ is a regular distribution, whereas $\F_{\mathrm{disc}}$ is irregular. Refer to Figure~\ref{fig:reg_disc} for a geometric illustration of the revenue curves.

\begin{figure}[t]
\begin{subfigure}[t]{0.45\textwidth}
\includegraphics[width=\textwidth,keepaspectratio]{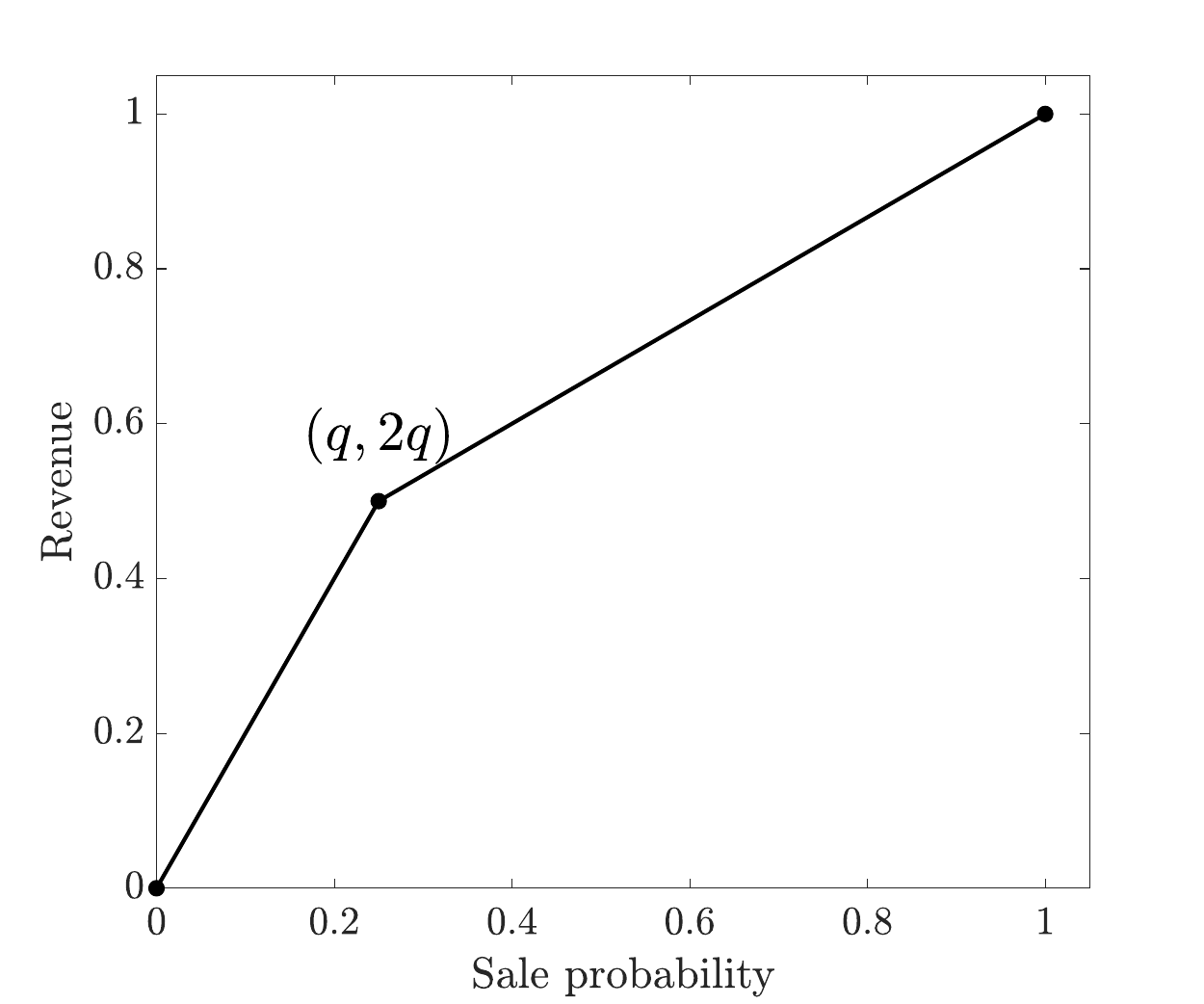}
\caption{$\F_{\mathrm{reg}}$}
\end{subfigure}
\begin{subfigure}[t]{0.45\textwidth}
\includegraphics[width=\textwidth,keepaspectratio]{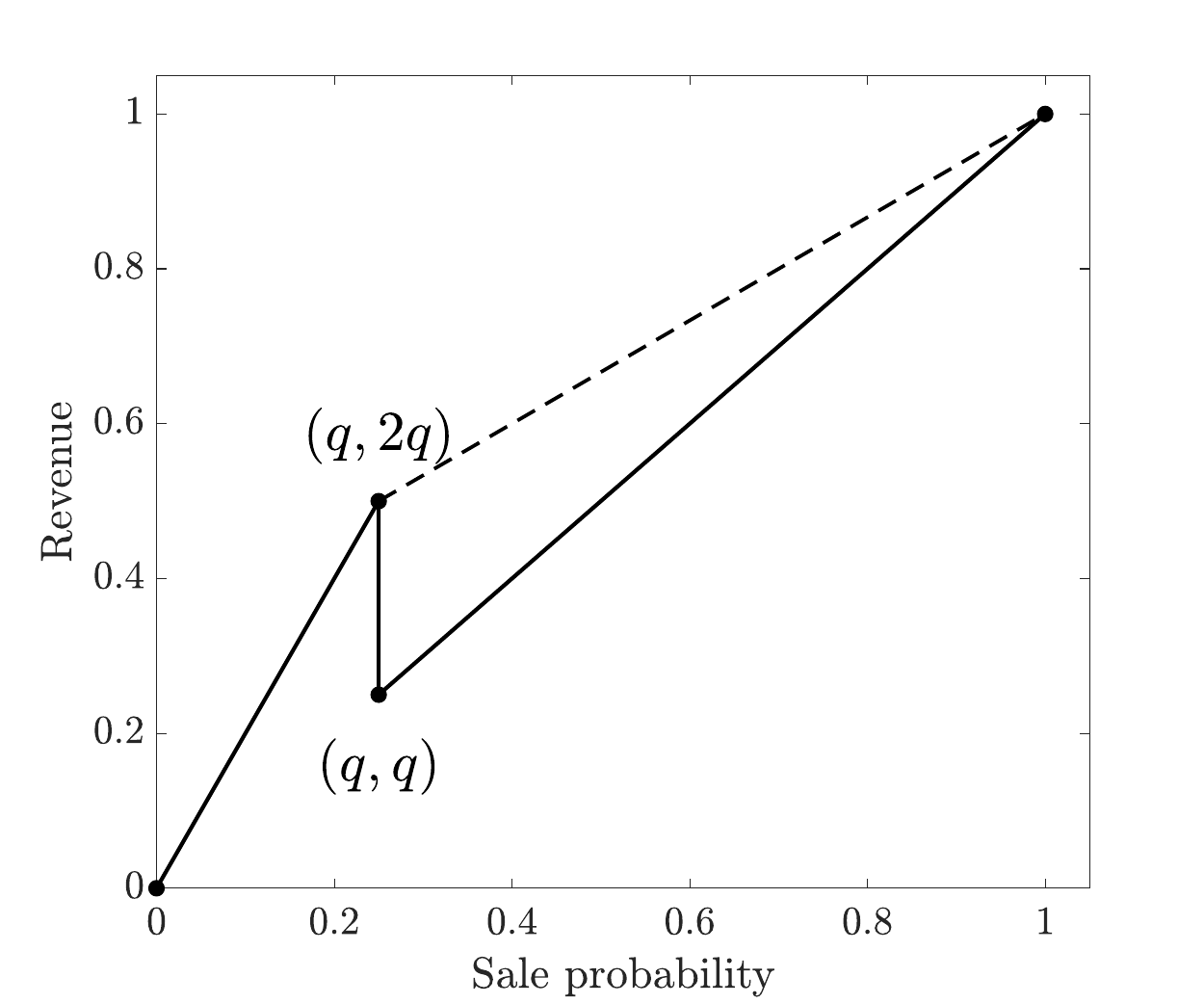}
\caption{$\F_{\mathrm{disc}}$}
\end{subfigure}
\centering
\captionsetup{justification=centering}
\caption{Revenue Curves of $\F_{\mathrm{reg}}$ and $\F_{\mathrm{disc}}$. The solid lines represent the revenue curves of $\F_{\mathrm{reg}}$ and $\F_{\mathrm{disc}}$. Ironing corresponds to taking the concave envelope of this curve, which is illustrated by the dashed line. By construction, the concave envelope of $\F_{\mathrm{disc}}$ coincides with the revenue curve of $\F_{\mathrm{reg}}$.}
\label{fig:reg_disc}
\end{figure}

Finally, consider a setting with three buyers. Let $\G_{\mathrm{disc}}$ (resp. $\G_{\mathrm{reg}}$) be the distribution of the second order statistic of three i.i.d.\ draws from $\F_{\mathrm{disc}}$ (resp. $\F_{\mathrm{reg}}$).

\begin{proposition}
We have the following:
\begin{itemize}
    \item $\R[2](\G_{\mathrm{disc}}) < \opt(\F_{\mathrm{disc}}^3)$;
    \item We have
    \[
    \R[2](\G_{\mathrm{reg}}) = \min_{\dists \in \Pi_2(\G_{\mathrm{reg}})} \Rev(\spa(0), \dists) >\min_{\dists \in \Pi_2(\G_{\mathrm{reg}})} \Rev(\myerson(\F_{\mathrm{reg}}^3), \dists),
    \]
    where $\myerson(\F_{\mathrm{reg}}^3)$ denotes the implementation of Myerson's auction that pools all values in $[1,2)$, described in the proof.
\end{itemize}
\end{proposition}

The first statement suggests that the optimal revenue on $n$ i.i.d.\ $\Fiid[2](\G)$ might not be achievable when $\Fiid[2](\G)$ does not satisfy regularity above the reserve. Notice that this example is minimal in the following sense.
(1) With only two buyers, the second highest value coincides with the smallest value. In this case, Theorem~\ref{thm:last} implies that Myerson’s auction is robustly optimal.
(2) If the distribution $\G$ has support size one, the problem becomes degenerate and admits a trivial solution. Our construction $\G_{\mathrm{disc}}$ has support size $2$.

The second statement shows that for certain distributions,  the second-price auction with a reserve is robustly optimal, while Myerson's auction, despite being optimal for the i.i.d. distribution $\Fiid[2](\G)^n$, fails to achieve the same worst-case revenue guarantee. In other words, second price auction is more robust.

\begin{proof}

We start with the first statement and compute the optimal revenue for $\F_{\mathrm{disc}}^3$.
Since $\bar{\varphi}_{\mathrm{disc}}(1) = 2-1/q > 0$ and $\bar{\varphi}_{\mathrm{disc}}(2) = 2$, and since the optimal revenue equals the expected maximum of the (ironed) virtual values and zero, we have
\[
\opt(\F_{\mathrm{disc}}^3) = \Ex[\vals \sim \F_{\mathrm{disc}}^3]{\max_{i \in [3]} \bar{\varphi}_{\mathrm{disc}}(v_i)} = 2 - q^2~.
\]
One feasible implementation of Myerson's auction is to allocate the item to the highest bidder and to charge the winner the minimum winning bid, under the deterministic lexicographic tie-breaking, where bids in $[1,2)$ are treated as the bid $1$.

A subtle distinction between Myerson's auction and the standard second-price auction arises when the valuation profile $\vals$ is $(1,2,1)$ or $(1,1,2)$. In these cases, Myerson's auction charges the winner a payment of $2$, whereas the second-price auction charges $1$, interpreting the winning bid as $1+\varepsilon \to 1$.
Taking this into account, the expected optimal revenue can as well be calculated in the following way:
\begin{align*}
\opt(\F_{\mathrm{disc}}^3) &= 2 \cdot \Big(
\underbrace{(1-q)^3 + 3(1-q)^2 q}_{\text{at least two bids equal }2}
+ \underbrace{2(1-q)q^2}_{\vals \in \{(1,2,1),(1,1,2)\}} \Big) + 1 \cdot \Big( \underbrace{(1-q)q^2}_{\vals=(2,1,1)} + \underbrace{q^3}_{\vals = (1,1,1)} \Big) \\
& = 2 - q^2~.
\end{align*}

Next, we construct an alternative independent distribution whose second order statistic distribution is consistent with $\G_{\mathrm{disc}}$, but yields strictly smaller optimal revenue.
Notice that the distribution $\G_{\mathrm{disc}}$ of the second order statistic $v_{(2)}$ is supported on two points:
\[
v_{(2)} =
\begin{cases}
1 & \text{with probability } 3q^2 - 2q^3, \\
2 & \text{otherwise}.
\end{cases}
\]
Let $\F_1 = \F_2$ be a two-point distribution that assigns value $1$ with probability
\[
1 - \sqrt{\,1 - 3q^2 + 2q^3\,},
\]
and value $2$ otherwise. Let the third buyer be a dummy bidder with deterministic value $0$, i.e., $\F_3 = \delta_0$.
By direct inspection, when $\vals \sim \F_1 \times \F_2 \times \F_3$, the induced distribution of the second order statistic $v_{(2)}$ is exactly $\G_{\mathrm{disc}}$.

Since the third buyer never affects allocation or payments, the seller's optimal revenue coincides with that in the two-buyer setting $\F_1 \times \F_2$. Moreover, the ironed virtual value function of $\F_1$ is
\[
\bar{\varphi}_{\F_1}(v) = 2-\frac{1}{1-\sqrt{1-3q^2+2q^3}} \ge 0, \qquad \forall v \in [1,2)
\]
Here the inequality holds by our assumption on $q$. Hence, the expected optimal revenue can be calculated as the following:
\begin{align*}
\R[2](\G_{\mathrm{disc}})
&\le \opt(\F_1 \times \F_2 \times \F_3) = \opt(\F_1 \times \F_2) \\
& = 2 \cdot \big( \P[v_1=v_2=2] + \P[v_1=1,v_2=2] \big)
+ 1 \cdot \big( \P[v_1=2,v_2=1] + \P[v_1=v_2=1] \big) \\
&= 1 + \P[v_2=2] = 1 + \sqrt{\,1 - 3q^2 + 2q^3\,}.
\end{align*}
Finally, notice that for every $q \in (0,1)$,
\[
1 + \sqrt{\,1 - 3q^2 + 2q^3\,} < 2 - q^2,
\]
which implies
\[
\R[2](\G_{\mathrm{disc}}) < \opt(\F_{\mathrm{disc}}^3).
\]
This establishes that solving for a consistent i.i.d.\ distribution and running Myerson's auction is not robustly optimal when only the distribution of the second order statistic is known.

Next, we study the second statement. Since $\F_{\mathrm{reg}}$ is a regular distribution (its revenue curve is concave), Theorem~\ref{thm:second} implies that the second-price auction is robustly optimal for $\G_{\mathrm{reg}}$. Here the monopoly price of $\F_{\mathrm{reg}}$ is $p^*=1$, because its revenue curve is increasing; and since $v_{(2)}\ge1$ almost surely under every $\dists\in\Pi_2(\G_{\mathrm{reg}})$, the reserve $1$ never binds, so $\spa(1)$ and $\spa(0)$ earn the same revenue under every such $\dists$. In particular, given only the distribution $\G_{\mathrm{reg}}$ of the second order statistic, the second-price auction maximizes the seller's worst-case expected revenue:
\[
\R[2](\G_{\mathrm{reg}}) = \min_{\dists \in \Pi_2(\G_{\mathrm{reg}})} \Rev(\spa(0), \dists) = \opt(\F_{\mathrm{reg}}^3)~.
\]

On the other hand, because the virtual value function is constant on $[1,2)$, there exists an alternative optimal mechanism for $\F_{\mathrm{reg}}^3$. In particular, Myerson's auction can be implemented by pooling all values in $[1,2)$ together. More formally, consider the following auction:
\begin{itemize}
    \item Treat all bids in the interval $[1,2)$ as identical, and allocate the item to the bidder with the highest bid, breaking ties lexicographically (or uniformly at random); bids below $1$, which occur with probability zero under $\F_{\mathrm{reg}}$, are rejected.
    \item Charge the winner the minimum winning bid.
\end{itemize}
This is the same as Myerson's auction designed for $\F_{\mathrm{disc}}^3$, i.e.,
\[
\myerson(\F_{\mathrm{reg}}^3) = \myerson(\F_{\mathrm{disc}}^3)~.
\]
Notice that this pooled version of Myerson's auction achieves the same optimal revenue as the second-price auction on $\F_{\mathrm{reg}}^3$, i.e., $\Rev(\myerson(\F_{\mathrm{reg}}^3), \F_{\mathrm{reg}}^3) = \opt(\F_{\mathrm{reg}}^3)$. Moreover, $\varphi_{\mathrm{reg}}(v_i)$ equals $2$ with probability $1-q$ and $2-1/q$ otherwise, exactly as $\bar\varphi_{\mathrm{disc}}(v_i)$ under $\F_{\mathrm{disc}}$, so $\opt(\F_{\mathrm{reg}}^3)=\opt(\F_{\mathrm{disc}}^3)=2-q^2$.

We remark that by slightly perturbing the revenue curve of $\F_{\mathrm{reg}}$ to make it non-concave, the pooling over the interval $[1,2)$ becomes the unique implementation of Myerson’s auction up to tie-breaking. Specifically, keep the endpoints $(1-q,2(1-q))$ and $(1,1)$ of the linear piece of the revenue curve fixed and move the curve strictly below this chord in between. Then the atom at $2$, the value $\G(2^-)=3q^2-2q^3$ of the induced distribution of the second order statistic, the pooled auction, and the ironed virtual values are all unchanged. The argument below therefore shows that the robust performance guarantee of Myerson's auction is still at most $1+\sqrt{1-3q^2+2q^3}$, which is strictly below its revenue $2-q^2$ on the perturbed i.i.d.\ distribution.
In contrast, the second-price auction with a reserve price incurs only a negligible loss in its performance guarantee due to the perturbation, because by Theorem~\ref{thm:second} its guarantee equals its revenue on the perturbed i.i.d.\ distribution, which converges to $\Rev(\spa(0),\F_{\mathrm{reg}}^3)=2-q^2$ as the perturbation vanishes.

We now examine its robust performance. Denote the pooled auction above by $\mech$. On profiles whose values lie in $\{0\}\cup[1,2]$, its allocation and payments depend on the values only through the classes $\{0\}$, $[1,2)$, and $\{2\}$, and every critical bid lies in $\{1,2\}$. Hence $\Rev(\mech,\dists)=\Rev(\mech,\hat{\dists})$ whenever $\dists$ is supported on $(\{0\}\cup[1,2])^3$ and $\hat{\dists}$ is obtained from $\dists$ by mapping each value in $[1,2)$ to $1$. Let $\mathbb{H}:=1-\sqrt{1-\G_{\mathrm{reg}}}$ and $\dists^\star:=\mathbb{H}\times\mathbb{H}\times\delta_0$. Under $\dists^\star$, the second-highest value is the minimum of two independent draws from $\mathbb{H}$, whose distribution function is $1-(1-\mathbb{H})^2=\G_{\mathrm{reg}}$; thus $\dists^\star\in\Pi_2(\G_{\mathrm{reg}})$. Since $\mathbb{H}$ is supported on $[1,2]$ and
\[
\Pr[V\sim\mathbb{H}]{V\ge2}=\sqrt{1-\G_{\mathrm{reg}}(2^-)}=\sqrt{1-3q^2+2q^3},
\]
the mapping above transforms $\dists^\star$ into the distribution $\F_1\times\F_2\times\F_3$ constructed in the first part of the proof. Therefore,
\begin{align*}
\min_{\dists \in \Pi_2(\G_{\mathrm{reg}})} \Rev(\mech,\dists)
&\le \Rev(\mech,\dists^\star) = \Rev(\mech,\F_1\times\F_2\times\F_3)
\le \opt(\F_1\times\F_2\times\F_3) \\
&< \opt(\F_{\mathrm{disc}}^3) = \opt(\F_{\mathrm{reg}}^3) = \R[2](\G_{\mathrm{reg}})~,
\end{align*}
where the strict inequality was established in the first part.
\end{proof}

\paragraph{Second Price Auction with Reserve vs.\ Myerson's Auction.} Finally, when $\Fiid[2](\G)$ is not regular above the reserve, the performance of these two mechanisms on the consistent i.i.d.\ distribution $\Fiid[2](\G)^n$ still provides natural lower and upper bounds for the robust optimization problem:
\[
\Rev(\spa(p^*), \Fiid[2](\G)^n) \le \R[2](\G) \le \Rev(\myerson(\Fiid[2](\G)^n), \Fiid[2](\G)^n),
\]
where the first inequality follows from Theorem~\ref{thm:second} and the second inequality arises from the optimality of $\myerson$ under $\Fiid[2](\G)^n$. In Section~\ref{sec:universal}, we show that the second-price auction with a suitably chosen reserve, which can be computed from $\Fiid[2](\G)^n$ alone, guarantees at least a $2-\sqrt{2}\approx0.586$ fraction of $\R[2](\G)$, up to an arbitrarily small additive loss, for every $\G$ and without any regularity assumption.
Given the simplicity of the second-price auction with a reserve and its guaranteed improvement over the second-price auction without a reserve, we argue that this mechanism is well justified in practice. Indeed, as illustrated by the above example, the second-price auction without a reserve may already be robustly optimal, and switching to Myerson’s auction can even reduce the seller’s worst-case revenue.

\subsection{A Universal Guarantee for the Second-Price Auction}
\label{sec:universal}

When $\Fiid[2](\G)$ is regular above its reserve, Theorem~\ref{thm:second} shows that the second-price auction with a reserve is robustly optimal, and Section~\ref{sec:example} shows that exact robust optimality of a natural candidate mechanism can fail without such a condition. We now show that, without any condition on $\G$, the second-price auction with a suitably chosen reserve is approximately robustly optimal. Its best robust revenue is
\[
\mathcal{S}_2(\G):=\sup_{p\ge0}\,\inf_{\dists\in\Pi_2(\G)}\Rev\bigl(\spa(p),\dists\bigr)
=\sup_{p\ge0}\,\Rev\bigl(\spa(p),\Fiid[2](\G)^n\bigr),
\]
where the equality follows from Theorem~\ref{thm:second}. Thus, the seller can choose the reserve by optimizing the second-price auction for the consistent i.i.d.\ distribution, which is determined by $\G$ alone.

\begin{theorem}\label{thm:two-minus-sqrt-two}
For every $n\ge2$ and every distribution $\G$ on $[0,\infty)$,
\[
\mathcal{S}_2(\G)\ \ge\ \bigl(2-\sqrt{2}\bigr)\,\R[2](\G)\ \approx\ 0.586\,\R[2](\G).
\]
The inequality holds in the extended nonnegative reals, without regularity, continuity, or moment assumptions on $\G$. In particular, if $\R[2](\G)<\infty$, then for every $\varepsilon>0$, some reserve $p\ge0$ satisfies
\[
\inf_{\dists\in\Pi_2(\G)}\Rev\bigl(\spa(p),\dists\bigr)\ \ge\ \bigl(2-\sqrt{2}\bigr)\,\R[2](\G)-\varepsilon.
\]
If $\G$ has bounded support, the supremum defining $\mathcal{S}_2(\G)$ is attained, and one can take $\varepsilon=0$.
\end{theorem}

The proof compares both sides of the inequality with the posted-price revenue of the consistent i.i.d.\ distribution,
\[
P:=\sup_{p\ge0}\;p\,\Bigl(1-\Fiid[2](\G)(p^-)^n\Bigr).
\]
On the one hand, the second-price auction with reserve $p$ sells at a price of at least $p$ whenever $v_{(1)}\ge p$, so Lemma~\ref{lem:second_first} yields $\mathcal{S}_2(\G)\ge P$. On the other hand, the ambiguity set contains an instance with only two active bidders on which no mechanism earns more than $(1+1/\sqrt{2})\,P$. The constant comes from the following scalar inequality.

\begin{lemma}\label{lem:scalar}
Fix $n\ge2$, let $c:=1/\sqrt{2}$, and let $h_n(u):=nu^{n-1}-(n-1)u^n$. For every $u\in[0,1]$, the quantities
\[
q:=\sqrt{1-h_n(u)}\qquad\text{and}\qquad\beta:=1-u^n
\]
satisfy
\begin{equation}\label{eq:scalar}
q\ \le\ \beta\,\bigl[c+(1-c)\,q\bigr].
\end{equation}
\end{lemma}

\begin{proof}
For $u=1$, both sides of \eqref{eq:scalar} vanish. Fix $u\in[0,1)$ and define
\[
b:=\sqrt{2}-1=\frac{1-c}{c},\qquad
U(u):=\sum_{j=0}^{n-1}u^j,\qquad
W(u):=\sum_{j=0}^{n-2}(j+1)\,u^j.
\]
Then $\beta=(1-u)\,U(u)$ and $1-h_n(u)=(1-u)^2\,W(u)$, so $q=(1-u)\sqrt{W(u)}$. Substituting these expressions into \eqref{eq:scalar}, rearranging, and dividing by $c\,(1-u)>0$ shows that \eqref{eq:scalar} is equivalent to $\sqrt{W(u)}\,(1+bu^n)\le U(u)$. As both sides are nonnegative, it suffices to prove that
\[
D(u):=U(u)^2-W(u)\,(1+bu^n)^2\ \ge\ 0.
\]
Write $D(u)=\sum_{j}d_ju^j$. Expanding $U(u)^2$ and $W(u)(1+bu^n)^2=W(u)+2bu^nW(u)+b^2u^{2n}W(u)$ shows that $d_j=0$ for $j\le n-2$ and for $j=2n-1$, that $d_{n-1}=n$, and that, for $0\le m\le n-2$,
\[
d_{n+m}=n-1-m-2b(m+1)
\qquad\text{and}\qquad
d_{2n+m}=-b^2(m+1).
\]
Since $d_{n+m}$ is strictly decreasing in $m$, the nonzero coefficients change sign exactly once, from positive to negative. Hence there is an index $j_0$ such that $d_j\ge0$ for $j\le j_0$ and $d_j\le0$ for $j>j_0$. For $0<u<1$, this gives $d_ju^j\ge d_ju^{j_0}$ for every $j$, and therefore
\[
D(u)\ \ge\ u^{j_0}\sum_jd_j\ =\ u^{j_0}D(1)\ =\ u^{j_0}\Bigl(n^2-\frac{n(n-1)}{2}(1+b)^2\Bigr)\ =\ n\,u^{j_0}\ >\ 0,
\]
where we used $(1+b)^2=2$. Finally, $D(0)=U(0)^2-W(0)=0$.
\end{proof}

\begin{proof}[Proof of Theorem~\ref{thm:two-minus-sqrt-two}]
Let $\F:=\Fiid[2](\G)$ and $c:=1/\sqrt2$. By Claim~\ref{clm:iid-inversion}, $\G(v)=h_n(\F(v))$ for every $v$, where $h_n=H_{n,2}$ is the polynomial in Lemma~\ref{lem:scalar}. Since $h_n$ is continuous and strictly increasing on $[0,1]$, also $\G(v^-)=h_n(\F(v^-))$ for every $v$.

\smallskip\noindent\textit{Step 1: $\mathcal{S}_2(\G)\ge P$.}
Fix $p\ge0$ and $\dists\in\Pi_2(\G)$. Whenever $v_{(1)}\ge p$, the auction $\spa(p)$ sells the item at a price of at least $p$. By Lemma~\ref{lem:second_first}, $\OS[1](\dists)\succeq\OS[1](\F^n)$, so
\[
\Rev\bigl(\spa(p),\dists\bigr)\ \ge\ p\,\Pr[\vals\sim\dists]{v_{(1)}\ge p}\ \ge\ p\,\Pr[\vals\sim\F^n]{v_{(1)}\ge p}\ =\ p\,\bigl(1-\F(p^-)^n\bigr).
\]
Taking the infimum over $\dists$ and then the supremum over $p$ gives $\mathcal{S}_2(\G)\ge P$. If $P=\infty$, the theorem follows. Henceforth, assume $P<\infty$.

\smallskip\noindent\textit{Step 2: an adversarial instance with two active bidders.}
Define the distribution function $\mathbb{H}(v):=1-\sqrt{1-\G(v)}$ and
\[
\dists^\star:=\mathbb{H}\times\mathbb{H}\times\delta_0\times\cdots\times\delta_0.
\]
Under $\dists^\star$, the second-highest value is the minimum of two independent draws from $\mathbb{H}$, whose distribution function is $1-(1-\mathbb{H})^2=\G$. Thus $\dists^\star\in\Pi_2(\G)$. We claim that, for every $p\ge0$, both $q=\Pr[V\sim\mathbb{H}]{V\ge p}$ and $q=\Pr[V\sim\mathbb{H}]{V>p}$ satisfy
\begin{equation}\label{eq:threshold}
p\,q\ \le\ P\,\bigl[c+(1-c)\,q\bigr].
\end{equation}
For the closed threshold, let $u:=\F(p^-)$. Then $q=\sqrt{1-\G(p^-)}=\sqrt{1-h_n(u)}$, so Lemma~\ref{lem:scalar} and the definition of $P$ give
\[
p\,q\ \le\ p\,(1-u^n)\bigl[c+(1-c)\,q\bigr]\ \le\ P\bigl[c+(1-c)\,q\bigr].
\]
For the open threshold, let $u:=\F(p)$, so that $q=\sqrt{1-\G(p)}=\sqrt{1-h_n(u)}$. The same argument applies because $p\,(1-\F(p)^n)=\lim_{p'\downarrow p}p'\,\bigl(1-\F(p'^-)^n\bigr)\le P$ by the right continuity of $\F$.

\smallskip\noindent\textit{Step 3: no mechanism earns more than $(1+c)P$ on $\dists^\star$.}
Fix a DSIC and IR mechanism $\mech$. By IR, a bidder who reports zero pays at most zero, i.e., $t_j(0,\vals_{-j})\le0$ for every bidder $j$ and every $\vals_{-j}$. Subtracting $t_j(0,\vals_{-j})$ from all of bidder $j$'s payments, for each fixed $\vals_{-j}$, preserves DSIC because the shift does not depend on $j$'s own report, preserves IR because, by Myerson's payment identity, bidder $j$'s utility becomes $\int_0^{v_j}x_j(s,\vals_{-j})\dd s\ge0$, and weakly increases revenue. We may therefore assume that every bidder who reports zero pays zero. Fix an active bidder $i\in\{1,2\}$ and the values $\vals_{-i}$ of the other bidders, and let $x_i(v)\in[0,1]$ and $t_i(v)$ denote bidder $i$'s expected allocation and payment when she reports $v$. By DSIC and Myerson's payment identity, $x_i$ is nondecreasing and
\[
t_i(v)=v\,x_i(v)-\int_0^v x_i(s)\dd s.
\]
Write $x_i(v)=\int_0^1\ind[v\in L_z]\dd z$, where $L_z:=\{v\ge0:x_i(v)\ge z\}$. Since $x_i$ is nondecreasing, each $L_z$ with $z\in(0,1]$ is empty or of the form $[p_z,\infty)$ or $(p_z,\infty)$ for some $p_z\ge0$. In both nonempty cases, $v\,\ind[v\in L_z]-\int_0^v\ind[s\in L_z]\dd s=p_z\,\ind[v\in L_z]$; set $p_z:=0$ when $L_z$ is empty. By Tonelli's theorem,
\[
t_i(v)=\int_0^1 p_z\,\ind[v\in L_z]\dd z,
\]
so the payment rule is a mixture of threshold payments of the form $p\,\ind[v\ge p]$ or $p\,\ind[v>p]$. Conditional on $\vals_{-i}$, bidder $i$'s value is distributed according to $\mathbb{H}$ by independence. For each $z$, \eqref{eq:threshold} gives $p_z\Pr{v_i\in L_z\mid\vals_{-i}}\le P\bigl[c+(1-c)\Pr{v_i\in L_z\mid\vals_{-i}}\bigr]$, which holds trivially when $L_z$ is empty. Integrating over $z$ (all integrands are nonnegative) yields
\[
\E\bigl[t_i(v_i)\mid\vals_{-i}\bigr]
=\int_0^1 p_z\,\Pr{v_i\in L_z\mid\vals_{-i}}\dd z
\le c\,P+(1-c)\,P\,\E\bigl[x_i(v_i)\mid\vals_{-i}\bigr].
\]
The remaining bidders have value zero and therefore pay zero. Summing over $i\in\{1,2\}$, taking expectations, and using the feasibility constraint $x_1+x_2\le1$ gives
\[
\Rev(\mech,\dists^\star)\ \le\ 2cP+(1-c)\,P\,\E\bigl[x_1(v_1)+x_2(v_2)\bigr]\ \le\ (1+c)\,P.
\]
Since $\dists^\star\in\Pi_2(\G)$ and $\mech$ was arbitrary,
\[
\R[2](\G)\ \le\ \sup_{\mech}\Rev(\mech,\dists^\star)\ \le\ (1+c)\,P\ \le\ (1+c)\,\mathcal{S}_2(\G).
\]
Since $(1+c)^{-1}=2-\sqrt2$, this proves the inequality, and the statement about $\varepsilon$ follows from the definition of the supremum. Finally, by the proof of Theorem~\ref{thm:second}, $\Rev(\spa(p),\F^n)=p\,(1-\F(p^-)^n)+\int_p^\infty(1-\G(v))\dd v$. This function of $p$ is upper semicontinuous, because $p\mapsto\F(p^-)$ is left continuous and $\F(p)\ge\F(p^-)$; and it vanishes for $p$ above the support of $\G$. Hence its supremum is attained when $\G$ has bounded support.
\end{proof}

Theorem~\ref{thm:two-minus-sqrt-two} shows that the simple and widely used second-price auction with a reserve retains a constant fraction of the optimal robust revenue even when our exact optimality guarantee does not apply. Combined with Theorem~\ref{thm:second}, the reserve can be computed by optimizing over the consistent i.i.d.\ distribution, so the guarantee requires no information beyond $\G$. Determining the tight approximation ratio remains an open question.

\subsection{Extensions to General Intermediate Order Statistics and Applications}
\label{sec:k-statistic}

We prove Lemma~\ref{lem:general-order-statistic} for general $k$ and $i$ and Proposition~\ref{prop:nonexistence}, and discuss applications to multi-unit and position auctions.

\subsubsection{Proof of Lemma~\ref{lem:general-order-statistic}}
We consider $n$ independent random variables $v_1,v_2,\dots,v_n$ with (unknown) marginal distributions $\dists=\F_1 \times \F_2 \times \dots \times \F_n$. Let $v_{(k)}$ denote the $k$-th order statistic, i.e., the $k$-th largest value among $v_1,\dots,v_n$. We are given the distribution function $\G$ of $v_{(k)}$.

Fix a value $v\in\mathbb{R}$. We consider the following optimization problem:
\begin{align*}
\max_{\dists}:\quad
& \Pr[\vals\sim\dists]{v_{(i)} \le v}\\
\text{subject to}:\quad
& \Pr[\vals\sim\dists]{v_{(k)} \le v} = \G(v)
\end{align*}
This is an \emph{order-statistic constrained optimization} problem: among all choices of marginal distributions that realize the prescribed distribution of the $k$-th order statistic, we seek the one that minimizes the distribution of the $i$-th order statistic.

A useful way to view the problem is to consider the induced binary random variables $X_j=\ind[v_j>v]$. Under a product distribution, the $X_j$'s are independent Bernoulli random variables, and the events $\{v_{(t)}\le v\}$ are completely determined by the sum $\sum_j X_j$. Thus, for each fixed $v$, the problem reduces to a finite-dimensional optimization over Bernoulli parameters $\{x_j\}$, where $x_j:=\Pr{v_j>v}$.
And we have
\[
\textstyle \Pr{v_{(t)} \le v} = \Pr{\sum_j X_j < t} \qquad \text{for every } t\in[n].
\]
Hence, for each fixed $v$, the optimization problem reduces to
\[
\begin{aligned}
\max_{\mathbf{x}\in[0,1]^n}: \quad & \textstyle \Pr{\sum_j X_j <i}\\
\text{subject to}: \quad & \textstyle \Pr{\sum_j X_j < k} = \G(v).
\end{aligned}
\]

An easy case is when $i>k$, as the objective is a probability, it is at most $1$.
On the other hand, take any feasible product distribution realizing the constraint $\Pr{\sum_j X_j<k}=\G(v)$ using only $k$ coordinates, and set the remaining $n-k$ coordinates to satisfy $x_j=0$ (i.e., $X_j \equiv 0$ a.s.). Then $\Pr{\sum_j X_j<i} = 1$ for $i>k$. This establishes the second part of the statement.

Next, we focus on the case $i<k$. Let $y:=1-\Fiid(\G)(v)$, so that the i.i.d.\ parameter vector $\mathbf y:=(y,\dots,y)$ is feasible; since $H_{n,k}$ is strictly increasing (Claim~\ref{clm:iid-inversion}), it is the unique feasible vector with equal coordinates. Write $\P_{\mathbf x}$ for probabilities when $X_1,\dots,X_n$ are independent with $\P_{\mathbf x}[X_j=1]=x_j$, and let $S_n:=\sum_{j=1}^nX_j$. We prove that
\begin{equation}\label{eq:bernoulli-claim}
\P_{\mathbf x}[S_n<i]\le\P_{\mathbf y}[S_n<i]
\qquad\text{for every }\mathbf x\in[0,1]^n\text{ with }\P_{\mathbf x}[S_n<k]=\G(v).
\end{equation}
Assuming this claim holds, the lemma follows immediately from the fact that
\[
\OS[i]\left(\Fiid(\G)^n\right)(v) \ge \OS[i](\dists)(v)
\]
for every $v \in \mathbb{R}$, which is precisely the definition of the asserted
first-order stochastic dominance.

If $\G(v)=1$, then $y=0$ and the right-hand side of \eqref{eq:bernoulli-claim} equals one. If $\G(v)=0$, then the left-hand side is at most $\P_{\mathbf x}[S_n<k]=0$ because $i<k$. Hence, assume $0<\G(v)<1$. The feasible set $\mathcal X:=\{\mathbf x\in[0,1]^n:\P_{\mathbf x}[S_n<k]=\G(v)\}$ is nonempty and compact, and the objective $\mathbf x\mapsto\P_{\mathbf x}[S_n<i]$ is continuous, so the set $\mathcal X^*$ of its maximizers over $\mathcal X$ is nonempty and compact. It suffices to show that $\mathbf y\in\mathcal X^*$. We use a pairwise equalization step, stated for the pair of coordinates $1$ and $2$; the same argument applies to any pair.

Fix $\mathbf x\in\mathcal X$ and let $S := \sum_{j=3}^n X_j$. Then $S$ is independent of $(X_1,X_2)$, and for any integer $t$ we have
\begin{align}
\Pr{\sum_j X_j < t}
&=\Pr{S<t-2}+\Pr{S=t-2}\Pr{X_1+X_2\le 1}+\Pr{S=t-1}\Pr{X_1+X_2=0}\notag\\
&=\Pr{S<t-2}+\Pr{S=t-2}(1-x_1x_2)+\Pr{S=t-1}(1-x_1)(1-x_2)\notag\\
&=\Pr{S<t}-\Pr{S=t-1}(x_1+x_2)+\bigl(\Pr{S=t-1}-\Pr{S=t-2}\bigr)x_1x_2.
\label{eq:expand}
\end{align}

Write $w_\ell:=\P[S=\ell]$, $s:=x_1+x_2$, and $r:=x_1x_2$. Since $S$ is a sum of independent Bernoulli random variables, its probability mass function is log-concave and its support is an interval of integers; indeed, the generating polynomial $\prod_{j\ge3}(1-x_j+x_jz)$ has only real roots, so Newton's inequalities apply to its coefficients. Consequently,
\begin{equation}\label{eq:lc}
w_\ell\,w_m\ \ge\ w_{\ell-1}\,w_{m+1}\qquad\text{whenever }\ell\le m.
\end{equation}
To see this, note that if the right-hand side is positive, then $w_j>0$ for $\ell-1\le j\le m+1$, and log-concavity implies that $w_{j-1}/w_j$ is nondecreasing in $j$ on this range.
Suppose that $x_1\ne x_2$. We distinguish three cases.

\smallskip\noindent\textit{Case (a): $w_{k-1}>0$.} By \eqref{eq:expand} with $t=k$, the constraint is equivalent to $s=\lambda+\mu r$, for a constant $\lambda$ and $\mu:=1-w_{k-2}/w_{k-1}\le1$. Substituting this into \eqref{eq:expand} with $t=i$ shows that, along feasible pairs, the objective equals a constant plus $\kappa r$, where
\[
\kappa:=\frac{w_{i-1}\,w_{k-2}-w_{i-2}\,w_{k-1}}{w_{k-1}}\ \ge\ 0
\]
by \eqref{eq:lc} with $\ell=i-1\le m=k-2$. Let $\psi(\rho):=\lambda+\mu\rho-2\sqrt{\rho}$ for $\rho\in[0,1]$. Then $\psi(r)=s-2\sqrt r>0$ because $x_1\ne x_2$, and $\psi$ is strictly decreasing because $\mu\le1$. Moreover, $\psi(1)\le0$: the right-hand side of \eqref{eq:expand} with $t=k$ is strictly decreasing in $s$, and at $(s,r)=(2,1)$ it equals $\P[S+2<k]\le\P_{\mathbf x}[S_n<k]=\G(v)$, so that $\lambda+\mu\le2$. Hence, there is $\rho^*\in(r,1]$ with $\psi(\rho^*)=0$, and replacing both $x_1$ and $x_2$ by $\sqrt{\rho^*}$ preserves feasibility, strictly increases $x_1x_2$, and changes the objective by $\kappa(\rho^*-r)\ge0$.

\smallskip\noindent\textit{Case (b): $w_{k-1}=0<w_{k-2}$.} By \eqref{eq:expand} with $t=k$, feasibility pins down $r$. Replacing both $x_1$ and $x_2$ by $\sqrt r$ preserves feasibility and the product $x_1x_2$, strictly decreases $x_1+x_2$, and changes the objective by $w_{i-1}(s-2\sqrt r)\ge0$.

\smallskip\noindent\textit{Case (c): $w_{k-1}=w_{k-2}=0$.} This case does not arise in the two steps below.

\smallskip\noindent\textit{Step 1: some maximizer has no zero coordinate.} Let $\mathbf x\in\mathcal X^*$ have the fewest zero coordinates, and suppose, relabeling if necessary, that $x_2=0$. Because $\P_{\mathbf x}[S_n\ge k]=1-\G(v)>0$, at least $k$ coordinates of $\mathbf x$ are positive; relabel so that $x_1>0$. Then at least $k-1$ of the coordinates $x_3,\dots,x_n$ are positive, so $\P[S\ge k-1]>0$; moreover, $\P[S\le k-1]\ge\P_{\mathbf x}[S_n<k]=\G(v)>0$. Since the support of $S$ is an interval of integers, $w_{k-1}>0$, and Case~(a) produces a point of $\mathcal X^*$ with fewer zero coordinates, a contradiction.

\smallskip\noindent\textit{Step 2: some maximizer has equal coordinates.} By Step~1, $\max_{\mathbf x\in\mathcal X^*}\prod_jx_j>0$. Among the maximizers of $\prod_jx_j$ over $\mathcal X^*$, which form a nonempty compact set, choose $\hat{\mathbf x}$ minimizing $\sum_jx_j$; all coordinates of $\hat{\mathbf x}$ are positive. Suppose that $\hat x_1\ne\hat x_2$. Case~(c) cannot occur. Indeed, the support of $S$ is $\{L,\dots,n-2\}$, where $L$ is the number of indices $j\ge3$ with $\hat x_j=1$. If $k\le n-1$ and $w_{k-2}=w_{k-1}=0$, then $L\ge k$, so $S_n\ge k$ almost surely, contradicting $\G(v)>0$; if $k=n$, then $w_{k-2}=\prod_{j\ge3}\hat x_j>0$. In Case~(a), equalization yields a point of $\mathcal X^*$ with a strictly larger product; in Case~(b), it yields a point of $\mathcal X^*$ with the same product and a strictly smaller sum. Both contradict the choice of $\hat{\mathbf x}$. Hence, all coordinates of $\hat{\mathbf x}$ are equal, so $\hat{\mathbf x}=\mathbf y$. This proves \eqref{eq:bernoulli-claim} and completes the proof of Lemma~\ref{lem:general-order-statistic}.

\subsubsection{Proof of Proposition~\ref{prop:nonexistence}}
Suppose, for contradiction, that $\dists^*=\F_1\times\cdots\times\F_n\in\Pi_k(\G)$ is simultaneously minimal. Let $\tilde\F(v):=1-(1-\G(v))^{1/k}$ and $\dists^0:=\tilde\F^{\,k}\times\delta_0^{\,n-k}$. Under $\dists^0$, the $k$-th largest value is the minimum of $k$ independent draws from $\tilde\F$, whose distribution is $\G$; hence $\dists^0\in\Pi_k(\G)$, and $v_{(k+1)}=0$ almost surely under $\dists^0$. By \eqref{eq:simultaneous} with $i=k+1$ and nonnegativity, $v_{(k+1)}=0$ almost surely under $\dists^*$. If $k+1$ bidders had a positive probability of a positive value, independence would give $\P[v_{(k+1)}>0]>0$. Hence, at least $n-k$ bidders, say bidders $k+1,\dots,n$, have value zero almost surely under $\dists^*$.

Since $\G$ is not a point mass, there is $v\ge0$ with $0<\G(v)<1$; fix such $v$. Let $a_j:=\P_{\F_j}[v_j>v]$ for $j\le k$ and $g:=1-\G(v)\in(0,1)$. Because the other bidders have value zero, $\prod_{j\le k}a_j=\Pr[\vals\sim\dists^*]{v_{(k)}>v}=g$, and by the arithmetic--geometric mean inequality,
\[
\Pr[\vals\sim\dists^*]{v_{(1)}\le v}=\prod_{j\le k}(1-a_j)\ \le\ \Bigl(1-\frac1k\sum_{j\le k}a_j\Bigr)^k\ \le\ (1-a)^k,
\qquad\text{where } a:=g^{1/k}\in(0,1).
\]
We now exhibit Bernoulli parameters that satisfy the constraint in \eqref{eq:bernoulli-claim} and attain a strictly larger value of $\P[S_n<1]$. The function $\phi(\xi):=a(2\xi-\xi^2)+(k-1)(1-a)\xi^2$ is continuous, with $\phi(0)=0<a^2$ and $\phi(a)=a^2+k(1-a)a^2>a^2$. Hence, there is $\xi\in(0,a)$ with $\phi(\xi)=a^2$. Since $n\ge k+1$, we may consider the parameter vector $\mathbf x$ with $k-1$ coordinates equal to $a$, two coordinates equal to $\xi$, and $n-k-1$ coordinates equal to $0$. Under $\mathbf x$, the event $\{S_n\ge k\}$ occurs if and only if either all coordinates equal to $a$ succeed and at least one coordinate equal to $\xi$ succeeds, or exactly $k-2$ coordinates equal to $a$ succeed and both coordinates equal to $\xi$ succeed. Hence,
\[
\P_{\mathbf x}[S_n\ge k]=a^{k-1}(2\xi-\xi^2)+(k-1)a^{k-2}(1-a)\xi^2=a^{k-2}\phi(\xi)=a^k=g,
\]
so $\P_{\mathbf x}[S_n<k]=\G(v)$. Moreover, $\phi(\xi)=a^2$ and $\xi>0$ imply $a(2\xi-\xi^2)<a^2$, i.e., $(1-\xi)^2>1-a$, and therefore
\[
\P_{\mathbf x}[S_n<1]=(1-a)^{k-1}(1-\xi)^2>(1-a)^k.
\]
By \eqref{eq:bernoulli-claim} with $i=1$, the i.i.d.\ distribution satisfies
\[
\Pr[\vals\sim\Fiid(\G)^n]{v_{(1)}\le v}=\P_{\mathbf y}[S_n<1]\ \ge\ \P_{\mathbf x}[S_n<1]\ >\ (1-a)^k\ \ge\ \Pr[\vals\sim\dists^*]{v_{(1)}\le v}.
\]
Thus, $\OS[1](\dists^*)\not\preceq\OS[1](\Fiid(\G)^n)$, which contradicts \eqref{eq:simultaneous} with $i=1$.

\subsubsection{Applications}
Beyond its theoretical interest, Lemma~\ref{lem:general-order-statistic} for general $i$ and $k$ has important implications for practical mechanism design.
In many widely used mechanisms, the realized outcomes, especially the transfers, depend only on a small number of leading order statistics of bidders' valuations.
In Section~\ref{sec:second}, we establish the robustness of the second-price auction with a reserve by exploiting the fact that its revenue depends only on the distributions of the first- and second-order statistics.
We generalize this observation by defining a class of mechanisms whose total revenue depends only on the highest $k$ valuations. For this broad class of monotone and symmetric mechanisms, we prove that their worst-case revenue guarantee is attained in the i.i.d.\ case. Consequently, if the revenue-maximizing auction under the i.i.d.\ distribution admits an implementation that depends only on the top $k$ order statistics, then it is robustly optimal within our framework.

\begin{definition}
\label{def:k-order-statistic}
We say that a mechanism $\mech$ is a \emph{top-$k$ order-statistic mechanism} if its total payment depends only on the top $k$ order statistics of the valuation profile $\vals$ and is separable across these coordinates. Formally, there exist non-decreasing functions $f_j : \mathbb{R} \to \mathbb{R}$ for each $j \in [k]$ such that, for every valuation profile $\vals$,
\[
\sum_{i \in [n]} p_i(\vals) = \sum_{j \in [k]} f_j\left(v_{(j)}\right).
\]
We use $\mathcal{T}_k$ to denote the class of all top-$k$ order-statistic mechanisms.
\end{definition}

\begin{theorem}
    For every $\mech \in \mathcal{T}_k$ and every $\ell \ge k$, we have
    \[
    \min_{\dists \in \Pi_\ell(\G)} \Rev(\mech, \dists) = \Rev(\mech, \Fiid[\ell](\G)^n)
    \]
\end{theorem}
\begin{proof}
Let $\{f_j\}_{j \in [k]}$ be the corresponding payment functions as defined in Definition~\ref{def:k-order-statistic}. Then we have
\begin{align*}
\Rev(\mech, \dists) & = \Ex[\vals \sim \dists]{\sum_{j \in [k]} f_j(v_{(j)})} = \sum_{j \in [k]} \Ex[v \sim \Phi_j(\dists)]{f_j(v)} \\
& \ge \sum_{j \in [k]} \Ex[v \sim \Phi_j(\bar{\F}_\ell(\G)^n)]{f_j(v)} = \Rev(\mech, \Fiid[\ell](\G)^n)~,
\end{align*}
where the inequality follows by Lemma~\ref{lem:general-order-statistic} and the monotonicity of $f_j$ (for $j=\ell$, which is possible only if $j=k=\ell$, the distribution $\OS[j](\dists)=\G$ does not depend on $\dists$). All expectations are well defined, because each $f_j$ is bounded from below on $[0,\infty)$.
\end{proof}
This theorem implies that any top-$k$ order-statistic mechanism designed for the consistent i.i.d.\ distribution is robust over the entire ambiguity set with $\ell$-th order statistic information for $\ell\ge k$. In particular, the second-price auction with a reserve belongs to the class $\cT_2$, since its total payment equals $p\,\ind[v_{(1)}\ge p]+(v_{(2)}-p)^+$, while Myerson's auction need not belong to $\cT_k$ for any $k<n$ when ironing is required (Section~\ref{sec:example}), which implies the robustness of $\spa$. Finally, we provide two other applications of top-$k$ order statistic mechanisms as a generalization of our previous results.

\paragraph{Multi-Unit Auction.} Consider a multi-unit auction where the seller has $k$ homogeneous units and each buyer is unit-demand. The seller allocates the units to the first $k$ bidders with highest bids and all winners pay a price of the $(k+1)$-th highest bid $v_{(k+1)}$. Then the seller could collect the distribution of the $(k+1)$-th order statistic. Now the seller can adopt this information to set a reserve price $r$ to improve his revenue, i.e., all units are sold at a price of $\max\offf{v_{(k+1)}, r}$.
The total revenue in this mechanism is determined entirely by the realization of the top $k+1$ bids.
\[
\sum_{j\in [n]}p_j(\vals) = r\cdot \sum_{j\in[k]} \mathbbm{1}[v_{(j)}\ge r] + k \cdot (v_{(k+1)}-r)^+ ~,
\]
which is increasing in each order statistic $v_{(j)}$ for $j=1,\dots, k+1$ and separable across these coordinates.
Consequently, if the seller wishes to optimize the reserve price $r$ given only information about the $(k+1)$-th order statistic, our theorem suggests that the minimum expected revenue is attained under the i.i.d.\ distribution consistent with this order-statistic information. Therefore, the robustly optimal reserve can be computed by solving the seller’s revenue-maximization problem under the i.i.d.\ assumption.

Alternatively, instead of setting a reserve price, the seller may choose
to optimize the number of units sold. Specifically, the seller may solve
\[
j^* = \argmax_{j \in [k]}\Ex{j \cdot v_{(j+1)}}~,
\]
and sell only $j^*$ units via a $(j^*+1)$-price auction.
That is, rather than allocating all $k$ units, the seller deliberately
restricts supply and sells only $j^*$ units.
This mechanism also belongs to the class $\mathcal{T}_{k+1}$, and our theorem
again implies that the worst-case expected revenue is attained under the i.i.d.\ distribution consistent with this order-statistic information.

\paragraph{Position Auction.}
Position auctions are foundational to the sponsored search markets operated by platforms such as Google and Bing. Consider a setting with $k$ advertising slots with click-through rates $\alpha_1\ge \alpha_2\ge \dots\ge\alpha_{k}$. The laddered auction proposed in \cite{aggarwal2006truthful} allocates slot $i$ to the bidder with the $i$-th highest value $v_{(i)}$. The payment charged to the $i$-th highest bidder is
\[
p_{(i)}(\vals) = \sum_{j=i}^{k} (\alpha_j-\alpha_{j+1})\cdot v_{(j+1)}
\]
where $\alpha_{k+1}=0$. Therefore, the total payment in this mechanism depends only on the first $k+1$ order statistics, so the laddered auction also belongs to the class of \emph{top-$(k+1)$ order-statistic mechanisms}.
More generally, consider a laddered auction with a reserve price $r\ge 0$.
In this case, bidder $i$ is assigned a slot only if $v_{(i)}\ge r$ and her payment is
\[
p_{(i)}(\vals) = \mathbbm{1}\offf{v_{(i)}\ge r}\cdot
\sum_{j=i}^{k} (\alpha_j-\alpha_{j+1})\cdot \max\offf{v_{(j+1)}, r}.
\]
Since $v_{(j+1)}>r$ implies $v_{(i)}>r$ for all $i\le j$, and since $\sum_{j=i}^{k}(\alpha_j-\alpha_{j+1})=\alpha_i$, the total payment collected by the seller equals
\[
\sum_{i=1}^{k} r\,\alpha_i\,\ind[v_{(i)}\ge r] \;+\; \sum_{j=1}^{k} j\,(\alpha_j-\alpha_{j+1})\,(v_{(j+1)}-r)^+,
\]
which is separable across the first $k+1$ order statistics and nondecreasing in each of them. Consequently, the laddered auction with reserve also belongs to the class of top-$(k+1)$ order-statistic mechanisms.
The seller can therefore determine the revenue-optimal reserve price by solving the problem under the assumption that bidders’ valuations are i.i.d., which achieves the optimal performance against all distributions consistent with the observed $(k+1)$-th order statistic.

\section{Conclusion and Open Questions}
In this work, we propose a framework for robust auction design under order-statistic information and characterize robust mechanisms corresponding to different order statistics.
We conclude by discussing several natural directions for future research.

\begin{itemize}
\item For the setting in which an intermediate order statistic is known, we establish the robustness of the second-price auction with a reserve and prove its conditional optimality. A natural open question is to characterize the optimal mechanism in full generality for this setting.

Furthermore, when the second order statistic is known, Theorem~\ref{thm:two-minus-sqrt-two} shows that the second-price auction with a suitably chosen reserve is a $(2-\sqrt{2})$-approximation to the robust mechanism design problem, up to an arbitrarily small additive loss, for every distribution $\G$ and with no regularity assumption. Given the simplicity and practical relevance of this mechanism, it is of interest to determine its tight worst-case approximation ratio, and to extend such guarantees to general intermediate order statistics. This direction fits into the theme of simple versus optimal mechanisms, where simple mechanisms are provably approximately optimal~\citep{hartline2009simple,alaei2019optimal,jin2019tight,jin2020tight}.

\item Another common approach in robust mechanism design is to evaluate mechanisms via approximation ratios, defined as the worst-case ratio between the revenue generated by a mechanism and the optimal revenue under full distributional knowledge \citep{azar2012optimal, azar2013optimal}.
Designing optimal approximate mechanisms when only the first order statistic is known is a promising yet technically challenging direction, as it requires moving beyond the degenerate instances considered in our framework. Notably, this approach does not readily extend to higher order statistics: for $k \ge 2$, it is possible that a single buyer with a sufficiently large valuation dominates the optimal revenue, while such information is not revealed by the $k$-th order statistic.

\item Throughout this work, we assume that the distribution of a single order statistic is known. A natural extension is to study settings in which the marginal or joint distributions of multiple order statistics are available. A particularly relevant example is the case where the joint distribution of $(v_{(1)}, v_{(2)})$ is known, which can be reliably collected from sealed-bid second-price auctions.

\item Finally, motivated by growing concerns over privacy, anonymity and non-discrimination, we advocate modeling informational constraints directly at the level of available data without imposing explicit restrictions on the mechanism itself.
An important open direction is to explore robust mechanism design under richer forms of anonymous and aggregated data, beyond order-statistic information.
\end{itemize}

\end{document}